\documentclass[11pt]{article}
\usepackage{graphicx} 
\usepackage[a4paper,left=2.3cm,right=2.3cm,top=2.5cm,bottom=4cm,]{geometry}

\usepackage{amsmath}
\usepackage{amsfonts}
\usepackage{amssymb}
\usepackage{color}
\usepackage{amsthm}
\usepackage{enumerate}

\newcommand{\der}[2]{\frac{\partial {#1}}{\partial {#2}}}
\newcommand{\var}[3]{\frac{\delta_{#1} {#2}}{\delta {#3}}}\newcommand{\cL}{\mathcal{L}}
\newcommand{\cA}{\mathcal A}
\newcommand{\cB}{\mathcal B}

\newcommand{\cF}{\mathcal F}
\newcommand{\cG}{\mathcal G}
\newcommand{\cJ}{\mathcal J}
\newcommand{\cK}{\mathcal K}
\newcommand{\cM}{\mathcal M}
\newcommand{\cN}{\mathcal N}

\newcommand{\cR}{\mathcal R}

\newcommand{\fX}{\mathfrak X}

\newcommand{\R}{\mathbb R}

\newcommand{\pair}[1]{\ensuremath{\left\langle#1\right\rangle}}

\DeclareMathOperator{\dd}{d\!}
\DeclareMathOperator{\im}{Im}

\usepackage[numbers]{natbib}
\usepackage{hyperref}

\usepackage{todonotes}

\usepackage{tikz-cd}
\tikzcdset{
	cells={font=\everymath\expandafter{\the\everymath\displaystyle}},
}

\newtheorem{theorem}{Theorem}
\newtheorem{lemma}[theorem]{Lemma}
\newtheorem{corollary}[theorem]{Corollary}
\newtheorem{proposition}[theorem]{Proposition}
\newtheorem*{remark}{Remark}

\newtheorem{definition}[theorem]{Definition}

\title{\vspace*{-10mm}\bf Duality of Hamiltonian and Lagrangian formulations {for integrable systems}}

\author{ Pierandrea Vergallo$^{1,2}$ \qquad  Mats Vermeeren$^{3,*}$
	\\[10pt]\small 
	$^1$Department of Basic and Applied Sciences, University of Basilicata,\\[1pt]\small 
	Via dell'Ateneo Lucano, 85100, Potenza, Italy\\[1pt]
	\small  $^2$Istituto Nazionale di Fisica Nucleare, 
	\small Sezione di Napoli, \\\small Via Cintia, 80126, Napoli, Italy\\[1pt]\small
	$^3$Department of Mathematical Sciences, 
	Loughborough University, \\[1pt] 
	\small Loughborough, Leicestershire LE11 3TU, United Kingdom}
\date{}

\begin{document}
	
	\maketitle
	
	{
		\renewcommand{\thefootnote}{\fnsymbol{footnote}}
		\footnotetext[1]{Corresponding author, e-mail: \texttt{m.vermeeren@lboro.ac.uk}}
		\enlargethispage{1em}
	}
	
	\begin{abstract}
		\noindent
		We introduce the concept of \emph{Hamiltonian potential variables} to map Hamiltonian operators into symplectic operators in a dual space. This generalises the classical trick of switching to a potential variable to obtain a Lagrangian density for the Korteweg-de Vries (KdV) equation. Building on this concept, we present the Lagrangian structure for bi-Hamiltonian systems, discuss the Lenard scheme in the symplectic formalisms, and apply this to construct pairs of Lagrangian multiforms. We discuss the key model of the KdV equation and some dispersionless limits of it. We present a pair of Lagrangian multiforms for these equations, one of which is new. We also consider the examples of polytropic gas dynamics and the constant astigmatism equation, for which no Lagrangian multiforms were previously known.
	\end{abstract}

\tableofcontents
\section{Introduction}

In the modern theory of integrable systems, Hamiltonian structures play a central role. The Lagrangian picture on the other hand, is not usually treated with much significance. This is in stark contrast to geometric mechanics, where the Hamiltonian and Lagrangian pictures are generally treated as equally important. Hamiltonian and Lagrangian formulations of partial differential equations (PDEs) are geometrically described by Poisson (or Hamiltonian) operators and symplectic operators, respectively. One reason for their unbalanced treatment lies in the fact that the integrability of a given evolutionary system is \emph{essentially} guaranteed by the existence of a bi-Hamiltonian structure, that is, by the presence of two compatible Hamiltonian formalisms describing the same system. The compatibility of two Hamiltonian operators is a simple linear criterion, namely that any linear combination of them is again a Hamiltonian operator. By contrast, in the Lagrangian framework a comparable linear test for the compatibility of symplectic structures is not, in general, available. 

In the last few decades, there have been several contributions to remedy the imbalance between Hamiltonian and Lagrangian perspectives in integrable systems. For example:

\begin{itemize}
\item It has been suggested to enforce the recurrence relation induced by a bi-Hamiltonian structure through a constrained variational principle \cite{dealmeidadasilva1990simple}.

\item There have been classifications of certain types of integrable systems, based on the existence of a Lagrangian density of a particular form \cite{ferapontov2006class, ferapontov2010integrable, ferapontov2010integrablea}.

\item Dirac structures, which generalise both symplectic and Hamiltonian operators, were thoroughly studied in \cite{dorfman1993dirac}.

\item Many examples of equations possessing a bi-Hamiltonian structure also possess different Lagrangian densities (that are not equivalent under the usual operations of adding total divergences). Such \emph{multi-Lagrangians} are studied in \cite{bustamante2003multilagrangians, nutku2001lagrangian, nutku2002multilagrangians, pavlov2017remarks}. We will refer to these structures as \emph{bi-Lagrangian structures}, in analogy to the term bi-Hamiltonian.

\item Similarly named, but fundamentally different, are \emph{Lagrangian multiforms} (and the closely related idea of \emph{pluri-Lagrangian systems}). Rather than having several Lagrangians for the same equation, a Lagrangian multiform describes a whole hierarchy of commuting equations. A Lagrangian multiform is a differential form: some of its coefficients can be understood as Lagrangian densities of the individual equations of a hierarchy, other coefficients have no clear meaning in the traditional calculus of variations \cite{lobb2009lagrangian, suris2016variational, suris2016lagrangian}.
The relation between Lagrangian multiforms and Poisson structures was studied in \cite{caudrelier2020hamiltonian, caudrelier2021multiform, vermeeren2021hamiltonian}, but no relation between Lagrangian multiforms and bi-Hamiltonian systems has previously been established.
\end{itemize}

In this paper, we combine the last two points. To reiterate, a bi-Lagrangian structure provides several Lagrangian structures for the \emph{same} equation, relating to its bi-Hamiltonian structure, whereas a Lagrangian multiform is a single object encoding a hierarchy of equations. Here, we combine both ideas. Given a bi-Hamiltonian hierarchy, we construct a pair of Lagrangian multiforms. This generalises both existing approaches, Lagrangian multiforms and bi-Lagrangians, to provide a pair of Lagrangian descriptions for the whole hierarchy, reflecting its bi-Hamiltonian structure.

Before treating Lagrangian multiform theory, we take a detailed look at the potential variables that are often introduced to turn a Hamiltonian partial differential equation (PDE) into a Lagrangian one. We argue that the transformation to a potential variable should be described using the Hamiltonian operator of the system. In the case of a compatible pair of Hamiltonian operators, transforming to a potential variable using one of the operators turns the pair into a pair of symplectic operators. This pair of symplectic operators leads to two Lagrangian structures (a bi-Lagrangian structure). Starting from each of these, a Lagrangian multiform can be constructed; thus we obtain two non-equivalent Lagrangian multiforms for the same bi-Hamiltonian hierarchy. Both Lagrangian multiforms have the system of evolutionary equations as their Euler-Lagrange equations. This is in contrast to the traditional Lagrangian formulation, where the Euler-Lagrange equations are differential consequences of the evolutionary equations.

The paper is organised as follows. In the remaining parts of the Introduction, we review the geometric setting of Hamiltonian and symplectic operators for evolutionary systems and the notion of compatibility of such operators.
In Section \ref{sec2}, starting from the motivating example of the Korteweg–de Vries (KdV) equation, we formalise the notion of potential coordinates and prove that, given a pair of compatible Hamiltonian operators, one can construct a corresponding pair of symplectic structures for the equation expressed in these new potential variables. We formulate a Lenard scheme associated with the resulting bi-symplectic structure, dual to its usual bi-Hamiltonian formulation. In Section \ref{multi-sec}, we briefly recall the main elements of the theory of Lagrangian multiforms and show how to construct a bi-Lagrangian 2-form starting from a bi-Hamiltonian structure. Section \ref{sec4} is devoted to examples. We first discuss the KdV equation, followed by several scalar models obtained as dispersionless limits of the potential KdV equation, as well as matrix examples arising in the theory of polytropic gas dynamics. Finally, we study the Constant Astigmatism Equation (CAE). We derive the bi-symplectic structure for the two-component form of CAE and present the resulting two Lagrangian 2-forms. We also present a Lagrangian multiform for the scalar form of the CAE. The paper concludes with some final remarks in the last section.

\subsection{Preliminary notions}
Let us consider a $2$-dimensional real manifold $M$, with local coordinates $t$ and $x$, a real $(n+2)$-dimensional manifold $E$, and a locally trivial vector bundle $\pi: E\rightarrow M$. We choose $u^1,\dots , u^n$ as local coordinates of the $n$-dimensional fibres $U$. These are the field variables of a system of $n$ evolutionary PDEs depending on $t$ and $x$, which are considered as the time and the space variables, respectively.
While it may be useful to have the geometric picture  of jet spaces on a vector bundle in mind (see e.g.\@ \cite{saunders1989geometry,krasilshchik2011geometry,bocharov1999symmetries}),
for our purposes it is sufficient to define $\mathbb{A}$ as the algebra of the differential functions $\varphi[u]$ of the space variable and higher order derivatives of ${u}=(u^1,\dots, u^n)$ with respect to $x$, i.e.\@ 
\[
     \varphi[u] = \varphi(x,{u},{u}_{x}, \dots, {u}_{kx}),
\]
 where ${u}_{kx}={u}_{\underbrace{x\cdots x}_{k \text{ times}}}$.
 
 The \emph{Fréchet derivative} of $\varphi$ is a 1-form on the space of sections of $\pi$ and acts on a vertical vector field $X = (X^1, \ldots, X^n)$ to $U$ by the local coordinate expression
 \begin{equation}
    \label{frechet}
     \ell_{\varphi} X = \sum_{k\geq 0} \der{\varphi}{u^j_{kx}} \partial^k_x X^j,
 \end{equation}
where the Einstein summation convention is assumed for repeated indices ranging from $1$ to $n$, while we choose to specify the infinite sums explicitly, and
 \[
     \partial_x = \der{}{x}+\sum_{k\geq 0}u^j_{(k+1)x}\der{}{u^j_{kx}}
 \]
 is the total derivative with respect to $x$. The formal adjoint of a Fréchet derivative is then given by the following formula:
 \begin{equation}
    \label{frechet-adjoint}
     \ell_\varphi^* X = \sum_{k\geq 0} (-1)^k \partial^k_x \left( \der{\varphi}{u^j_{kx}} X^j \right).
 \end{equation}
Finally, the Fréchet derivative of a 1-form $\varphi = \varphi_i \dd u^i$ is defined component-wise,
 \[
     \ell_{\varphi} X =  \ell_{\varphi_i} X \dd u^i .
 \]

Consider the quotient of $\mathbb{A}$ under the image of the total derivative $\partial_x$, i.e.\@ we consider differential functions which are equivalent up to total divergences. Elements in of this quotient are called functionals and are commonly indicated by $F=\int f \dd x$, where $f[u]$ is a representative in the class and is called a functional density.
We define the variational derivative of $F=\int f \dd x$ as the 1-form 
\[ \var{}{F}{u} = \var{}{F}{u^i} \dd u^i \]
with coefficients
\begin{equation}
    \label{varder}\var{}{F}{u^i} := \sum_{k\geq 0} (-1)^k \partial_{x}^k \der{f}{u_{kx}^i} . 
\end{equation} 
This expression does not depend on which representative $f$ in the equivalence class is considered, and when it is more natural to work directly with the density, we will also denote $\var{}{F}{u^i}$ by $\var{}{f}{u^i}$.
When dealing with scalar fields, we will make abuse of notation and will write $u$ for the first and only component $u^1$, and $\var{}{F}{u}$ for the coefficient $\var{}{F}{u^1}$.

In this context, an evolutionary vector field $X$ is a function from the $k$-th jet space of $\pi$ to the vertical subspace of its tangent space. It can be described in coordinates by
\[
    \varphi=\varphi^i(x,{u},\dots , {u}_{kx})\frac{\partial}{\partial u^i}.
\]
The space of all evolutionary vector fields is denoted by $\fX$, and its dual, the space of the corresponding 1-forms, by $\bar{\fX}$.

\paragraph{Hamiltonian and symplectic operators}

In this paper, we focus on (systems of) evolutionary PDEs
\begin{equation}\label{eveq}
    u^i_t=F^i({u},{u}_x,\dots , {u}_{kx}), \qquad i=1,2,\dots n,
\end{equation}
for some order $k\in\mathbb{N}$, where $F^i$ are smooth functions in their arguments.

We recall that \eqref{eveq} is called Hamiltonian if there exist a functional $H=\int h \dd x$ and a matrix differential operator $\cA: \bar \fX \to \fX$ with entries $\cA^{ij} = \sum_{\sigma \geq 0} a^{ij}_\sigma\partial_x^\sigma$, $a^{ij}_\sigma\in\mathbb{A}$, such that 
\begin{equation}
    \label{hamiltonian}
    u^i_t=\mathcal{A}^{ij}\left(\frac{\delta H}{\delta u^j}\right), \qquad i=1,2,\dots n,
\end{equation}
and
\begin{description}
    \item[A.1] $\mathcal{A}$ is skew-adjoint, i.e.\@ $\mathcal{A}^*=-\mathcal{A}$, and
    \item[A.2] the Schouten bracket of $\cA$ with itself vanishes,
    \[ [\cA,\cA](\varphi_1,\varphi_2,\varphi_3) =
    \pair{ \ell_{\cA} (\cA \varphi_1)(\varphi_2) , \varphi_3 } + \pair{ \ell_{\cA} (\cA \varphi_2)(\varphi_3) , \varphi_1 } + \pair{ \ell_{\cA} (\cA \varphi_3)(\varphi_1) , \varphi_2 }  
    = 0 , \]
    where 
    \[ \ell_\cA(X)(\varphi)^i = \sum_{\sigma,\tau \geq 0} \der{a^{ij}_\sigma}{u^k_{\tau x}} \partial_x^\tau X^k\, \partial_x^\sigma \varphi_j . \]
We refer to \cite{bocharov1999symmetries} for further discussions and for equivalent formulations of the present expression.
\end{description}
A differential operator $\cA$ is said to be Hamiltonian if A.1 and A.2 are satisfied. Equivalently, $\mathcal{A}$ is said Hamiltonian if the bracket
\[
    \{F,G\}_\mathcal{A}=\int \frac{\delta F}{\delta u^i}\mathcal{A}^{ij}\frac{\delta G}{\delta u^j} \dd x,
\]
 for functionals $F=\int f \dd x$, $G=\int g \dd x$, is a Poisson bracket, i.e. 
\begin{description}
    \item[B.1] it is skew-symmetric, $\{F,G\}_\cA=-\{G,F\}_\cA$, and
    \item[B.2] it satisfies the Jacobi identity 
    \[\{F,\{G,H\}_\cA\}_\cA + \{G,\{H,F\}_\cA\}_\cA + \{H,\{F,G\}_\cA\}_\cA = 0.\] 
\end{description}
We stress that A.1 and A.2 are equivalent to B.1 and B.2 respectively.

Dual to the notion of a Hamiltonian evolutionary system is the notion of a symplectic system, which is a system of the form
\[
    \cJ_{ij}(u^j_t)=\frac{\delta G}{\delta u^i}, \qquad i=1,2,\dots n,
\]
for some functional $G=\int g \dd x$ and a matrix differential operator $\cJ: \fX \to \bar \fX$ with entries $\cJ_{ij} = \sum_{\sigma \geq 0} b_{ij\sigma}\partial_{x}^\sigma$, $b_{ij\sigma}\in \mathbb{A}$, that satisfies
\begin{description}
    \item[C.1] $\cJ^*=-\cJ$, and
    \item[C.2] $ \pair{ \ell_{\cJ X^1} X^2, X^3 } +  \pair{ \ell_{\cJ X^2} X^3, X^1 }+  \pair{ \ell_{\cJ X^3} X^1, X^2 }=0$ for every $X^i\in \fX$.
\end{description}
We stress that a matrix differential operator satisfying the previous conditions is called symplectic. This is equivalent to the requirement that the 2-form $\omega$ defined by 
\[
    \omega(\xi,\eta)=\int \xi^i\cJ_{ij}\eta^j \dd x
\]
for $\xi,\eta\in\fX$, is a closed $2$-form.

In other words, Hamiltonian operators are matrix differential operators $\mathcal{A}:\bar{\fX}\rightarrow \fX$ whose associated bracket is a Poisson bracket, whereas  symplectic operators $\cJ:\fX\rightarrow \bar{\fX}$ are differential operators whose associate bilinear form is a closed $2$-form. 

It is important to note that condition C.2 is equivalent to the closedness of the differential operator $\cJ$. This implies a local exactness result which, for our purposes, reads as:
\begin{theorem}[{\cite[Chapter 6]{dorfman1993dirac}}]
    \label{thm-helm}
    Let $\cJ$ be a symplectic operator, then locally there exists a vertical $1$-form $p\in\bar{\fX}$ (or equivalently $n$ differential functions $p_i, i=1,2,\dots n$) such that 
    \begin{equation}\label{helm}
        \cJ=\ell_{p}-\ell^*_p,
    \end{equation}
    where $\ell$ is the Fréchet derivative and $\ell^*$ its formal adjoint.
\end{theorem}

In coordinates, condition \eqref{helm} reads
\[
    \cJ_{ij}=\displaystyle \sum_{k\geq 0}\left( \frac{\partial p_i}{\partial u^j_{kx}}\partial_x^k-(-1)^{k}\partial_x^k \circ \frac{\partial p_j }{\partial u^i_{kx}}\right) .
\]

The relevance of theorem \ref{thm-helm} is that operators of the form \eqref{helm} occur in Euler-Lagrange equations. Indeed, consider a Lagrangian density of the form $L = p_j u^j_t - h$, where $p_j$ and $h$ are functions of $u, u_x, u_{xx}$ but do not depend on $t$-derivatives. Its Euler-Lagrange equations take the form
\begin{align*}
    0 &= \var{}{L}{u^i} \\
    &= \sum_{k \geq 0} (-1)^k \partial_x^k \left( \der{p_j}{u^i_{kx}} u^j_t \right) - \partial_t p_i - \var{}{h}{u^i} \\
    &= \sum_{k \geq 0} (-1)^k \partial_x^k \left( \der{p_j}{u^i_{kx}} u^j_t \right) - \sum_{k \geq 0} \der{p_i}{u^j_{kx}} \partial_x^k u^j_t - \var{}{h}{u^i} \\
    &= -\cJ_{ij} u_t^j - \var{}{h}{u^i}.
\end{align*}
Hence, a system of PDEs is symplectic if and only if it is the system of Euler-Lagrange equations of a Lagrangian density of this form.

We finally remark that given a Hamiltonian operator $\cA$ that is nondegenerate, its inverse matrix operator $\cJ=\cA^{-1}$ is a symplectic operator. Vice versa, the inverse of an invertible symplectic structure is Hamiltonian. We refer to \cite{dorfman1993dirac,mokhov1998symplectic, mokhov2001symplectic} for further details.

\subsection{Compatibility for Hamiltonian and symplectic operators} 

Two Hamiltonian operators $\mathcal{A},\mathcal{B}$ are said to be compatible if every linear combination $\lambda \mathcal{A}+\mu \mathcal{B}$ is again a Hamiltonian operator \cite{magri1978simple}. An evolutionary system \eqref{eveq} is said to be bi-Hamiltonian if there exist two compatible Hamiltonian operators $\mathcal{A},\mathcal{B}$ such that it can be written as
\[
    u^i_t=\mathcal{A}^{ij}\frac{\delta H_1}{\delta u^j}=\mathcal{B}^{ij}\frac{\delta H_2}{\delta u^j}, \qquad i=1,2,\dots n.
\]
If, in addition, $\cA$ is invertible, then the recursion operator defined as $\cR = \cB \cA^{-1}$ is a Nijenhuis operator, i.e.\@ it satisfies 
\[     [ \mathcal{R} X , \mathcal{R} Y ] -  \mathcal{R}[ \mathcal{R} X , Y ] -\mathcal{R} [ X , \mathcal{R} Y ] + \mathcal{R}^2 [ X , Y ]=0, \qquad \forall X,Y\in {\fX} . \]

The notion of compatibility for symplectic operators is, in the general setting, not analogously natural. A notion of compatible symplectic structures is possible for fully arbitrary operators through the corresponding Dirac structures \cite[Section 3.6]{dorfman1993dirac}. A simpler definition is available in the case of a pair of symplectic operators $\cJ,\cK$ where $\cJ$ is non-degenerate. In this case, $\cJ$ and $\cK$ are said to be compatible if the recursion operator
\[
    \mathcal{R}=\cJ^{-1}\mathcal{K},
\]
is a Nijenhuis operator.
In case $\mathcal{J}$ and $\mathcal{K}$ are both non-degenerate, this condition is equivalent to the requirement that 
$\lambda \cJ^{-1}+\mu \cK^{-1}$ is a Hamiltonian operator for every $\lambda, \mu$, i.e.\@ that $\mathcal{J}^{-1}$ and $\mathcal{K}^{-1}$ are compatible as Hamiltonian structures. In analogy with the Hamiltonian framework, we then say that an evolutionary system \eqref{eveq} is {bi-symplectic} if there exist two compatible symplectic operators $\mathcal{J},\mathcal{K}$ and two functionals $G_1,G_2$ such that 
\[
    (\mathcal{J})_{ij}\, u^j_t=\frac{\delta G_1}{\delta u^i}\qquad \text{and}\qquad  (\mathcal{K})_{ij}\, u^j_t=\frac{\delta G_2}{\delta u^i}, \qquad i=1,2,\dots n,
\]
are differential consequences of the evolutionary system \eqref{eveq}.

 Following the previous definitions, it is evident that if a system is bi-Hamiltonian with non-degenerate Hamiltonian operators $\cA,\cB$, then it is also bi-symplectic with operators $\cJ = \cA^{-1}$ and $\cK = \cB^{-1}$. Indeed,
\begin{align*}
    &u^i_t=\cA^{ij}\frac{\delta H_1}{\delta u^j}\qquad \Longrightarrow \qquad \cJ_{ij}u^j_t=\frac{\delta H_1}{\delta u^j}, \qquad i=1,2,\dots , n,\\[5pt]
    &u^i_t=\cB^{ij}\frac{\delta H_2}{\delta u^j}\qquad \Longrightarrow \qquad \cK_{ij}u^j_t=\frac{\delta H_2}{\delta u^j}, \qquad i=1,2,\dots, n.
\end{align*}

Now, if there are two compatible Hamiltonian operators $\cA_0$ and $\cA_1$, then the recursion operator is $\cR = \cA_1 \circ \cA_0^{-1}$ and we (formally) have Hamiltonian operators $\cA_k  = \cR \circ \cA_{k-1}$ and symplectic operators 
$\cJ_k = \cA_k^{-1} = \cJ_{k-1} \circ \cR^{-1}$. If $\cJ_k$ is a local operator, we can write it as 
\begin{equation}
    \label{Ap}
    \cJ_k = \ell_{p_k} - \ell_{p_k}^*,
\end{equation}
where $p_k$ is a function of $u, u_x, u_{xx}, \ldots$. (For notational simplicity, we assume here that $u$ is a scalar field.) Then for each flow, indexed by $j$, we have a number of Lagrangians, indexed by $k$, given by
\[
    L = p_k[u] u_{t_j} - h_{j+k}[u] ,
\]
where $h_{j+k}[u]$ is a density of $H_{j+k}$. The corresponding Euler-Lagrange equation is
\[  
    \cJ_k {u}_{t_j} = - \frac{\delta H_{j+k}}{\delta {u}} .
\]

\section{Duality of formalisms and integrability}\label{sec2}

\subsection{Hamiltonian potential variables}\label{sec-pot}
      
The simplest example of a differential Hamiltonian operator is $\mathcal{A}=\partial_x$. This provides the first Poisson structure of the KdV equation:
\begin{equation*}
    u_t=6uu_x+u_{xxx}=\partial_x\,\frac{\delta}{\delta u}\left( u^3 - \frac12 u_x^2 \right).
\end{equation*}
In this case, one can consider a new variable $\bar{u}$ such that $u = \partial_x \bar u$. This new variable $\bar u$ is known as a potential variable. Applying this differential change of variables to the KdV equation and integrating with respect to $x$, we obtain the potential KdV equation
\begin{equation*}
    \bar{u}_t=3\bar{u}^2_x+\bar{u}_{xxx},
\end{equation*}
for which a Lagrangian is known.

It is no coincidence that the relation between $\bar u$ and $u$ is given by the Hamiltonian operator $\partial_x$. We will see below that the suitable transformation can be expressed in general (with abuse of notation) by $u = \cA \bar u$. To emphasise this, we will refer to $\bar u$ as the \emph{Hamiltonian potential variable}. It can also be thought of as a variable dual to the original $u$, because, geometrically, $\cA$ maps covectors into vectors.

\paragraph{Formal construction of Hamiltonian potential variables}

Consider two finite-dimensional vector spaces $U$, $\bar U$ that are dual to each other with a bilinear pairing $(\cdot,\cdot)$. They will both serve as possible spaces of dependent variables. As independent variable, take $x \in \R$. Now, as phase space we can use a space of sufficiently regular functions $\R \to U$ or $\R \to \bar U$. For example, we can take the usual conditions of rapid decay on $\cF$, and slightly more relaxed conditions on $\bar \cF$ to ensure that $\partial_x: \bar \cF \to \cF$ is surjective%
\footnote{
	This choice has the drawback of breaking the symmetry between $\cF$ and $\bar \cF$. One may be tempted to put different boundary conditions to restore this symmetry. For example, one could
	\begin{enumerate}
		\item leave out the boundary conditions, $\cF = \{ u: \R \to U \mid u \text{ smooth} \}$ and $\bar \cF = \{ \bar u: \R \to \bar U \mid \bar u \text{ smooth} \} $.
		
		This has the drawback that the pairing $\langle u, \bar v \rangle = \int (u(x),\bar v(x)) \dd x $ is not analytically defined, and that the operators $\partial_x: \cF \to \bar \cF$ and $\partial_x: \bar \cF \to \cF$ are surjective but not injective.

		\item restore some symmetry between $\cF$ and $\bar \cF$, and make $\partial_x$ a bijection, by setting
		\begin{align*}
			\cF = \{ u: \R \to U \mid \text{ smooth, rapidly decreasing at }- \infty, \text{ rapidly  approaching constant at }  + \infty \} , \\
			\bar \cF = \{ \bar u: \R \to \bar U \mid \text{ smooth, rapidly decreasing at } + \infty, \text{ rapidly  approaching constant at } - \infty \} .
		\end{align*}
		This has the drawback that $\var{}{H}{u}$ may not exist. For example, for $H[u] = u$ we have $\var{}{H}{u} = 1 \not\in \bar \cF$.
	\end{enumerate}
}%
:
\begin{align*}
&\cF = \{ u: \R \to U \mid \text{ smooth, rapidly decreasing at } \pm \infty \} , \\
&\bar \cF = \{ \bar u: \R \to \bar U \mid \text{ smooth, rapidly approaching constants at } \pm \infty \} .
\end{align*}
Then the pairing $(\cdot,\cdot)$ extends to a pairing $\langle\cdot,\cdot\rangle$ between $\cF$ and $\bar \cF$:
\[ \langle u, \bar v \rangle = \int (u(x),\bar v(x)) \dd x .\]
The spaces $\cF$ and $\bar \cF$ can be understood as smooth dual spaces. Their elements can be seen as sections of bundles $E \to \R$ and $\bar E \to \R$ with fibres $U$ and $\bar U$ respectively. 

Since $U$ and $\bar U$ are vector spaces, the (co)tangent bundles of $\cF$ and $\bar \cF$ are trivial bundles:
\[ T \cF \simeq \cF \times \cF, \qquad  T^* \cF \simeq \cF \times \bar \cF, \qquad  T \bar \cF \simeq \bar \cF \times \bar \cF, \qquad  T^* \bar \cF \simeq \bar \cF \times \cF. \]
In particular, we can think of vector fields on $\cF$ as maps $\cF \to \cF$, 1-forms on $\cF$ as maps $\cF \to \bar \cF$, etc.

Consider a functional $H: \cF \to \R$. A vector field $X: \cF \to \cF$ acts on the functional, creating a new functional $XH: \cF \to \R$ defined by
\[ (X H)[u] := \frac{\dd}{\dd \varepsilon} H[u + \varepsilon X[u]]\Big|_{\varepsilon=0},  \qquad u \in \cF. \] 
The variational derivative of $H: \cF \to \R$ is its differential, i.e.\@ the 1-form $\var{}{H}{u}: \cF \to \bar \cF$ satisfying:
\[ \pair{ \var{}{H}{u} , X } = X H , \qquad \text{i.e. } \pair{ \var{}{H}{u}[u] , X[u] } = (X H)[u] \]
for all vector fields $X: \cF \to \cF$. If $H$ is of the form $\int h \dd x$, then this definition agrees with the formula \eqref{varder}. For a functional  $H: \Bar \cF \to \cR$, we define $\var{}{\bar H}{\bar u}: \bar \cF \to \cF$ analogously.

In this setting, a Hamiltonian operator $\cA: \bar \fX \to \fX$ maps a function $\varphi: \cF \to \bar \cF$ to a function $\cF \to \cF$. We say that $\cA$ is constant (or has constant coefficients) if its action on the fibre does not depend on the base point, i.e.\@ if there is an operator $A: \bar \cF \to \cF$ such that $(\cA \varphi)[u] = A(\varphi[u])$.

\paragraph{Operators in Hamiltonian potential coordinates}

 We now investigate how this differential change of variables $u = A \bar u$ affects the \emph{nature} of $\cA$.

\begin{lemma}
    Consider a constant Hamiltonian operator $\cA: T^*\cF \to T \cF$ given fibre-wise by $A: \bar \cF \to \cF$ and a functional $H: \cF \to \R$. Then
    \[ \var{}{(H \circ A)}{\bar u}[\bar u] = -A \left(\var{}{H}{u}[u]\right) , \]
    where $u = A \bar u$.
\end{lemma}
\begin{proof}
    Take any $\bar X: \bar \cF \to \bar \cF$, $\bar u \in \bar \cF$, and let $u = A \bar u$:
    \begin{align*}
        \pair{  \var{}{(H \circ A)}{\bar u}[\bar u] , \bar X[\bar u] }
        & = (\bar X (H \circ A))[\bar u]       
        = \frac{\dd}{\dd \varepsilon} \Big|_{\varepsilon=0} H\big[A\big[\bar u + \varepsilon \bar X[\bar u] \big]\big] \\
        &= \frac{\dd}{\dd \varepsilon} \Big|_{\varepsilon=0} H\big[u + \varepsilon A \bar X[\bar u]\big] 
        = \pair{  \var{}{H}{u}[u] , A \bar X[\bar u] } 
    = \pair{ A^* \var{}{H}{u}[u] , \bar X[\bar u] } .
    \end{align*}
    The result follows because $A$ is skew-adjoint.
\end{proof}  

Note that the Hamiltonian potential variables are given by the formula
\[
    A:\bar{\cF}\longrightarrow \cF, \qquad \bar{u}^i\mapsto  {u}^i=A^{ij}\left(\frac{\delta \bar U}{\delta \bar u^j}\right), \qquad i=1,2,\dots n,
\]
where $\bar U$ is the functional 
\[ \bar U=\int{\sum_{l=1}^n\frac{(\bar u^l)^2}{2}\, dx}. \]
We will also refer to this mapping with the slight abuse of notation $u = \cA\bar u$. This transformation turns the evolutionary Hamiltonian system \eqref{hamiltonian} into
\begin{equation}
\label{bar-symplectic}
A_{ij}\,  \bar u_t^j = u_t^i = A^{ij}\left( \var{}{H}{u^j}\right) = - \var{}{H}{\bar u^i} , \qquad i=1,2,\dots n.
\end{equation}
This is a PDE in symplectic form, assuming that the constant operator $\bar \cA: T \bar \cF \to T^* \bar \cF$, given fibre-wise by $A: \bar \cF \to \cF$, is a symplectic operator. To show that this is indeed the case, it essentially suffices to observe that every constant skew symmetric operator is symplectic. (Similarly, every constant skew-symmetric operator is Hamiltonian \cite[Corollary 7.5]{olver1986applications}.) A more detailed statement, specialised to our current setting, is given in the following theorem:

\begin{theorem}
    \label{thm-Asymplectic}
    If $\cA: T^* \cF \to T \cF$ is a constant Hamiltonian operator, given fibre-wise by $A: \bar \cF \to \cF$, then the constant operator $\bar \cA: T \bar \cF \to T^* \bar \cF$, defined fibre-wise by the same $A: \bar \cF \to \cF$, is symplectic. 
\end{theorem}
\begin{proof}
	Consider $\omega \in \Omega^2(\bar \cF)$ defined by 
	\[ \omega\big|_{\bar u}(\bar X,\bar Y) = \pair{\bar X, A \bar Y} \]
	for $\bar u \in \bar \cF$ and $\bar X, \bar Y \in \bar \cF$ tangent vectors in the fibre at $\bar u$. We need to show that $\omega$ is closed. 
	
	Take arbitrary  $\bar u \in \bar \cF$ and $\bar X, \bar Y, \bar Z \in \bar \cF$. To compute $\dd \omega|_{\bar u} (\bar X, \bar Y, \bar Z)$, we extend $\bar X, \bar Y, \bar Z$ to constant vector fields $V_{\bar X}, V_{\bar Y}, V_{\bar Z}: \bar \cF \to \bar \cF$ given by $V_{\bar X}[\bar u] = \bar X$, etc. Then
	\begin{align*}
		\dd \omega (V_{\bar X}, V_{\bar Y}, V_{\bar Z})
		&= V_{\bar X} \omega(V_{\bar Y}, V_{\bar Z}) - \omega ([V_{\bar X}, V_{\bar Y}], V_{\bar Z}) + \text{cyclic permutations}.
	\end{align*}
	Observe that $\omega|_{\bar u}(V_{\bar Y}[\bar u], V_{\bar Z}[\bar u]) = \pair{\bar Y, A \bar Z}$ is independent of the base point $\bar u$, so $ V_{\bar X} \omega(V_{\bar Y}, V_{\bar Z}) = 0$. In addition, $[V_{\bar X}, V_{\bar Y}] = 0$ because the vector fields are constant. Hence $\dd \omega|_{\bar u} (\bar X, \bar Y, \bar Z) = \dd \omega|_{\bar u}  (V_{\bar X}[\bar u], V_{\bar Y}[\bar u], V_{\bar Z}[\bar u]) = 0$.
\end{proof}

\begin{remark}
	\label{remark-pullback}
	If $A: \bar \cF \to \cF$ is invertible, an alternative argument can be made, using the fact that the inverse of a Hamiltonian operator is symplectic. Consider $\bar \omega \in \Omega^2(\bar \cF)$ and  $\omega \in \Omega^2(\cF)$ given by
	\begin{align*}
		&\bar \omega(\bar X, \bar Y) = \pair{ \bar X, A \bar Y} \\
		&\omega(X, Y) = \pair{ A^{-1} \bar X, \bar Y} .
	\end{align*}
	These forms are related by the pullback $\bar \omega = A^* \omega$. From the fact that $\cA$ is Hamiltonian we conclude that $\cA^{-1}$ is symplectic, so $\omega$ is closed. Since the pullback of a closed form is closed, it follows that $\dd \bar \omega = 0$, so $\bar \cA$ is symplectic. 
\end{remark}

Note that both $\cA$ and $\bar \cA$ correspond to the same operator $\bar \cF \to \cF$ on the fibre. However, formally, the operator $\bar \cA$ is considered as a 2-form, in coordinates:
\[
    \mathcal{A}=\mathcal{A}^{ij}\,\,\frac{\delta}{\delta u^i}\wedge \frac{\delta}{\delta u^j} \qquad \text{and}\qquad \bar{\mathcal{A}}=\bar{\mathcal{A}}_{ij}\, \, d\bar u^i\wedge d\bar u^j.
\]
As a consequence, $\mathcal{A}$ acts on differential forms, whereas $\bar{\mathcal{A}}$ acts on vector fields. We stress that this is a generalisation of the corresponding classical case of Poisson tensors $\pi=(\pi^{ij})$ and symplectic forms $\omega=(\omega_{ij})$.

We finally conclude that if $\cA$ is invertible, constant and part of a bi-Hamiltonian pair, then every Hamiltonian operator in the hierarchy is symplectic with respect to the Hamiltonian potential variable $\bar u$:
\begin{theorem}
    \label{thm_1}
    Let $\cA: T^*\cF \to T \cF$ be an invertible constant Hamiltonian operator, given fibre-wise by $A: \bar \cF \to \cF$, and let $\cB: T^*\cF \to T\cF$ be a Hamiltonian operator compatible with $\cA$. Let $\bar \cA: T \bar \cF \to T^* \bar \cF$ be the constant symplectic operator defined by $A: \bar \cF \to \cF$ and define $\bar \cB: T \bar \cF \to T^* \bar \cF$ through its fibre-wise action $\cB[\bar u]: \bar \cF \to \cF$ by
    \[ \bar \cB[\bar u] := \cB[A \bar u] .\]
    Then $\bar \cA$, $\bar\cB$, and in general $\bar \cB (\bar\cA^{-1} \bar\cB)^{n}$ for $n\geq0$, are symplectic operators on $\bar \cF$.
\end{theorem}
\begin{proof}
    Compatibility implies that $(\cA + \lambda \cB): T^*\cF \to T\cF$ is Hamiltonian for all $\lambda$, so 
    \[ (\cA + \lambda \cB)^{-1} = \cA^{-1} - \lambda \cA^{-1} \cB \cA^{-1} + \lambda^2 \cA^{-1} \cB \cA^{-1} \cB \cA^{-1} - \ldots \]
    is symplectic on $\cF$ for all $\lambda$. In particular, this implies that $\cA^{-1} (\cB \cA^{-1})^{n}$ is symplectic on $\cF$ for all $n \geq 0$. 

    Similar to Remark \ref{remark-pullback}, we consider the 2-forms $\bar \omega_n \in \Omega^2(\bar \cF)$ and  $\omega_n \in \Omega^2(\cF)$ given by
    \begin{align*}
    	&\bar \omega_n|_{\bar u}(\bar X, \bar Y) = \pair{ \bar X, A (A^{-1} \bar \cB[\bar u])^n \bar Y} \\
    	&\omega_n|_{u}(X,Y) = \pair{ X, (A^{-1} \cB[u])^n A^{-1} Y} .
    \end{align*}
    We have
    \[ \omega_n|_{A \bar u}(A \bar X, A \bar Y) 
    = \pair{ A \bar X, (A^{-1} \cB[A \bar u])^n \bar Y}
    = -\pair{ \bar X, A (A^{-1} \bar \cB[\bar u])^n \bar Y} \]
    so the 2-forms are related by pullback,
    \[ A^* \omega_n = - \bar \omega_n . \]
    Hence $\dd \omega_n = 0$ implies $\dd \bar \omega_n = 0$, so $\bar \cA (\bar \cA^{-1} \bar \cB)^n: T \bar \cF \to T^* \bar \cF$ is a symplectic operator.
\end{proof}
Consequently: 
\begin{corollary}
    The operators $\bar \cA$ and $\bar \cB$ from Theorem \ref{thm_1} are compatible symplectic operators.
\end{corollary}
\begin{proof}
The symplectic structures defined by the operators in Theorem \ref{thm_1} can be written as $\omega_0(X,Y) = \pair{X, \bar \cA Y}$ and
\[ \omega_k(X,Y) = \omega_0(X, \bar \cR^k Y ), \]
where $\bar \cR = \bar \cA^{-1} \bar \cB $. They have the property that $\omega_k(\bar \cR X,Y) = \omega_k(X, \bar \cR Y)$. Using this, one can verify by direct computation that \cite[(3.14)]{dorfman1993dirac}
\begin{align*}
    \dd \omega_{k+1}(X,Y,Z) 
    &= \dd \omega_k(\bar \cR X,Y,Z) +  \dd \omega_k(X,\bar \cR Y,Z) - \dd \omega_{k-1}(\bar \cR X,\bar \cR Y,Z) \\
    &\qquad - \omega_0( [\bar \cR X, \bar \cR Y] - \bar \cR [\bar \cR X, Y] - \bar \cR[X, \bar \cR Y] + \bar \cR^2 [X,Y], \bar \cR^{k-1} Z ) .
\end{align*}
By Theorem \ref{thm_1}, all $\omega_k$ are closed, so it follows that
\[ \omega_0( [\bar \cR X, \bar \cR Y] - \bar \cR [\bar \cR X, Y] - \bar \cR[X, \bar \cR Y] + \bar \cR^2 [X,Y], \bar \cR^{k-1} Z ) = 0 . \]
Since $\bar \cA$ is invertible, the 2-form $\omega_0$ is non-degenerate. Hence, we see that $\bar \cR =  \bar \cA^{-1} \bar \cB$ is a Nijenhuis operator.
\end{proof}

\paragraph{Remarks on Dubrovin-Novikov operators}
The assumption that $\cA$ is constant is not very restrictive in physical contexts. Indeed, many examples in hydrodynamics, gas dynamics, biological systems, and mathematical physics in general, are described by Hamiltonian operators that can be made constant by a change of variables. A particular class of such Hamiltonian operators was introduced by Dubrovin and Novikov in 1983 \cite{dubrovin1983hamiltonianformalism}. These operators are homogeneous in the order of derivation. First-order homogeneous operators, as an example, take the following general structure
\begin{equation}\label{dn}
    g^{ij}\partial_x+b^{ij}_ku^k_x,
\end{equation}
where $g^{ij},b^{ij}_k$ depend on the field variables ${u}$ only. Under the non-degeneracy assumption on the leading coefficient, i.e.\@ $\det(g)\neq 0$, and with the additional requirements of A.1 and A.2 (the Hamiltonianity conditions) they are also known in the literature as Dubrovin-Novikov operators. Indeed, in \cite{dubrovin1983hamiltonianformalism} it was proved that, if $\det(g)\neq 0$, the operator \eqref{dn} is Hamiltonian if and only if $g_{ij}=(g^{lk})^{-1}$ is a flat metric and $b^{ij}_k=-g^{is}\Gamma^j_{sk}$, where $\Gamma^i_{jk}$ are the Christoffel symbols of the Levi-Civita connection of $g$.

Consequently, there always exist local coordinates such that a Dubrovin-Novikov operator \eqref{dn} can be mapped into the constant form
\[
    \eta^{ij}\partial_x, \qquad \eta^{ij}=\eta^{ji}\in\mathbb{R}.
\]
This change of variables takes the form $\tilde{u}^i=f^i({u})$, where $f^i$ are the Casimir functions of the operator and the new coordinates are known as flat coordinates for \eqref{dn}.

The case of bi-Hamiltonian pairs whose first operator is of Dubrovin-Nokivov type and their relations with the symplectic structures was investigated in \cite{pavlov2017remarks} by Pavlov and Vitolo. Their approach is to map $\mathcal{A}$ into constant form $\eta^{ij}\partial_x$, so that its inverse operator is simply given by $(\eta^{ij})^{-1}\partial_x^{-1}$. This operator is symplectic, so that if an evolutionary system has the form
\[
    u^i_t=F^i({u},{u}_x,\dots , {u}_{kx})=\eta^{ij}\partial _x\left(\frac{\delta h_1}{\delta u^j}\right),
\]
it also admit the non-local symplectic formalism
\[
    \eta_{ij}\partial_x^{-1}\, u^j_t=\frac{\delta h_1}{\delta u^j}.
\]
If $\cB$ is a second Hamiltonian operator compatible with $\cA$, then a bi-symplectic pair can be found as (see \cite{pavlov2017remarks} or \cite[Chapter 7]{dorfman1993dirac}):
\[
    \mathcal{J}=(\mathcal{A})^{-1}, \qquad \mathcal{K}= (\mathcal{A})^{-1}\circ \mathcal{B}\circ (\mathcal{A})^{-1}
\]

As Pavlov and Vitolo remarked in \cite{pavlov2017remarks}, with the non-local change of variables $u^i=\partial_x \bar u^i$, the non-locality of $\mathcal{J}$ and $\mathcal{K}$ is avoided and the resulting system becomes bi-symplectic with local operators. Up to a constant $\eta^{ij}$, this is the same as our definition $u = A \bar u$ of the Hamiltonian potential variable.
    
\subsection{Double Lenard scheme}

Let $\cA,\cB: T^* \cF \to T \cF$ be a Hamiltonian pair, with $\cA$ constant, given fibre-wise by $A: \bar \cF \to \cF$, define $\bar u$ by $u = A \bar u$.

\begin{figure}[ht]
    \centering
\begin{tikzcd}[column sep=large]
    \var{}{h_{0}}{u} \arrow[r, "\cA"] \arrow[rd, "\cB"] & u_{t_0} & \bar u_{t_0}  \arrow[l, " \cA"] \arrow[r, "\cJ = \bar \cA"] \arrow[rd, "\cK = \bar \cB"] & -\var{}{h_{0}}{\bar u} \\
    \var{}{h_{1}}{u} \arrow[r, "\cA"] \arrow[rd, "\cB"] & u_{t_1} & \bar u_{t_1}  \arrow[l, "\cA"] \arrow[r, "\cJ = \bar \cA"] \arrow[rd, "\cK = \bar \cB"] & -\var{}{h_{1}}{\bar u} \\
    \var{}{h_{2}}{u} \arrow[r, "\cA"] \arrow[rd, "\cB"] & u_{t_2} & \bar u_{t_2}  \arrow[l, "\cA"] \arrow[r, "\cJ = \bar \cA"] \arrow[rd, "\cK = \bar \cB"] & -\var{}{h_{2}}{\bar u} \\
    \vdots & \vdots & \vdots &  \vdots
\end{tikzcd}

$\underbrace{\hspace{3.5cm}}_{\text{Hamiltonian Lenard Scheme}}$
\hspace{1cm}
$\underbrace{\hspace{3.5cm}}_{\text{Symplectic Lenard Scheme}}$

\caption{Schematic overview of the two Lenard scheme, Hamiltonian and symplectic, linked by the Hamiltonian operator $\cA$.}
\label{fig-lenard}
\end{figure}
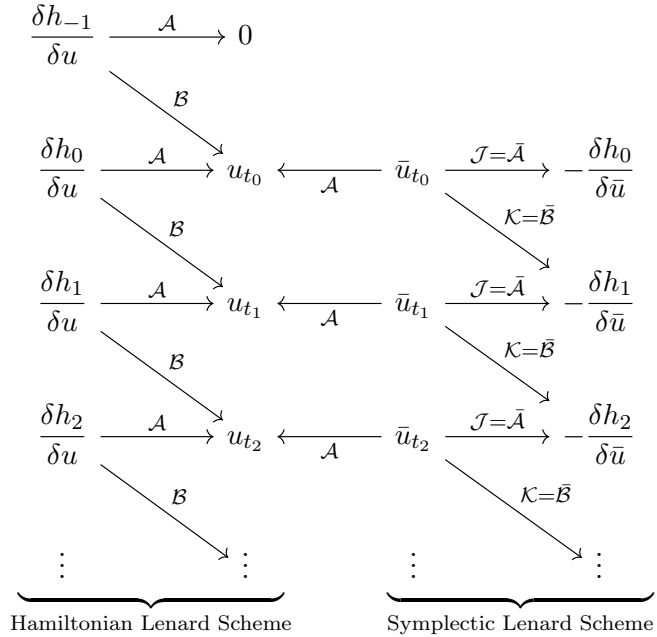

Suppose the equations of interest are $\bar u_{t_i} = Q_i$. Then the symplectic Lenard scheme, illustrated in Figure \ref{fig-lenard}, iterates the relations
\begin{align*}
    \bar \cA Q_i = - \var{}{h_i}{\bar u}, \\
    \bar \cB Q_i = - \var{}{h_{i+1}}{\bar u}
\end{align*}
Note that on the symplectic side, the recursion operator mapping $Q_i$ to $Q_{i+1}$ is given by $\bar \cA^{-1} \bar \cB$, whereas on the Hamiltonian side it is given by $\cB \cA^{-1}$.

A functional $H = \int h \dd x$ has exterior derivative $\dd H = \var{}{h}{\bar u^j} \dd \bar u^j$, so (the density of) a 1-form can be expressed locally as a variational derivative if and only if the 1-form is closed.

The following Lemma shows that each new 1-form that is produced by the iteration, $\phi_{i+1} = \bar\cB Q_i$, is closed. This means that it can be obtained as a variational derivative $\phi_{i+1} = \var{}{h_{i+1}}{\bar u^j} \dd \bar u^j$, hence we find a new Hamiltonian density $h_{i+1}$ and continue the iteration. 

\begin{lemma}[Special case of {\cite[Theorem 3.17]{dorfman1993dirac}}]
    Let $\cJ,\cK: T \cF \to T^* \cF$ be compatible symplectic operators, with $\cJ$ invertible. Assume that
    \[ \phi_0 = \cJ (Q_0) \qquad \text{and} \qquad \phi_1 = \cK (Q_0) = \cJ (Q_1) , \]
    where $\dd \phi_0 = \dd \phi_1 = 0$. Furthermore, assume that the set $V := \{Y \in T \cF \mid \cK Y \in \im \cJ \}$ is sufficiently large so that the only 2-form $\phi$ such that $\phi(Y,Z) = 0$ for all $Y,Z \in V$ is $\phi = 0$. Then $\phi_2 := \cK (Q_1)$ satisfies $\dd \phi_2 = 0$.
\end{lemma}
\begin{proof}
    From the assumptions it follows that $\cN = \cJ^{-1} \cK$ is a Nijhenhuis operator, i.e.\@ for all vector fields $Y$ and $Z$ there holds
    \begin{equation}
        \label{nijenhuis}
        [\cN Y, \cN Z] - \cN [ \cN Y, Z] - \cN [Y, \cN Z] + \cN^2[Y,Z] = 0.
    \end{equation}
    Choose any vector fields $Y_1,Z_1 \in V$. Then there exist $Y_2,Z_2$ such that
    \begin{align*}
        \cK Y_1 = \cJ Y_2 =: \phi_Y \qquad \text{and} \qquad
        \cK Z_1 = \cJ Z_2 .
    \end{align*}
    Then
    \[	\pair{ Q_0, \phi_Y } = \pair{ Q_0, \cJ Y_2 } = - \pair{ \cJ Q_0, Y_2 } = - \pair{ \phi_0, Y_2 } , \]
    but also
    \[	\pair{ Q_0, \phi_Y } = \pair{ Q_0, \cK Y_1 } = - \pair{ \cK Q_0, Y_1 } = - \pair{ \phi_1, Y_1 } , \]
    so 
    \begin{equation}
        \label{phiY}
        \pair{ \phi_0, Y_2 }  = \pair{ \phi_1, Y_1 } .
    \end{equation}
    Similarly, we find
    \begin{equation}
        \label{phiZ}
        \pair{ \phi_0, Z_2 }  = \pair{ \phi_1, Z_1 } .
    \end{equation}
    Next, we pair the Nijenhuis relation \eqref{nijenhuis} for $Y_1$ and $Z_1$ with $\phi_0$. Noting that $\pair{ \phi_i, \cN Y } = \pair{ \cK \cJ^{-1} \phi_i, Y } = \pair{ \phi_{i+1}, Y }$, we find
    \begin{equation}
        \label{nijenhuis2}
        \pair{ \phi_0, [Y_2,Z_2] } - \pair{ \phi_1, [Y_2, Z_1] } - \pair{ \phi_1, [Y_1, Z_2] } + \pair{ \phi_2, [Y_1,Z_1] } = 0 . 
    \end{equation}
    Now 
    \begin{align*}
        \dd \phi_0 (Y_2, Z_2) &= Y_2 \pair{ \phi_0, Z_2 } - Z_2 \pair{ \phi_0, Y_2 } + \pair{ \phi_0, [Y_2 ,Z_2] } , \\
        -\dd \phi_1 (U_1, Z_2) &= Y_1 \pair{ \phi_1, Z_2 } - Z_2 \pair{ \phi_1, Y_1 } - \pair{ \phi_1, [Y_1 ,Z_2] } , \\
        -\dd \phi_1 (Y_2, Z_1) &= Y_2 \pair{ \phi_1, Z_1 } - Z_1 \pair{ \phi_1, Y_2 } - \pair{ \phi_1, [Y_2 ,Z_1] } , \\
        \dd \phi_2 (Y_1, Z_1) &= Y_1 \pair{ \phi_2, Z_1 } - Z_1 \pair{ \phi_2, Y_1 } + \pair{ \phi_2, [Y_1 ,Z_1] } .
    \end{align*}
    Using Equations \eqref{phiY}--\eqref{nijenhuis2}, we see that
    \[ \dd \phi_0 (Y_2, Z_2) - \dd \phi_1 (Y_1, Z_2) - \dd \phi_1 (Y_2, Z_1) + \dd \phi_2 (Y_1, Z_1) = 0 , \]
    so the assumption that $\dd \phi_0 = \dd \phi_1 = 0$ implies that $ \dd \phi_2 (Y_1, Z_1) = 0$. Since $Y_1, Z_1 \in V$ are arbitrary, and $V$ is assumed to be sufficiently large, we conclude that $\dd \phi_2 = 0$.
\end{proof}

\begin{proposition}
    The $Q_i$ constructed this way are in involution: 
    \[ \pair{ Q_i, \cJ Q_j } = 0 \qquad \text{and}\qquad \pair{ Q_i, \cK Q_j } = 0. \]
\end{proposition}
\begin{proof}
    Similar to Equations \eqref{phiY}--\eqref{phiZ}, we obtain
    \[ \pair{ Q_i, \cJ Q_j } = \pair{ Q_i, \phi_j } = \pair{ Q_{i-1}, \phi_{j+1} } = \pair{ Q_{i-1}, \cJ Q_{j+1} } . \]
    So if $i-j$ is even, we find
    \[ \pair{ Q_i, \cJ Q_j } = \pair{ Q_{\frac{i+j}{2}}, \cJ Q_{\frac{i+j}{2}} } = 0 \]
    and if $i-j$ is odd, we find
    \[ \pair{ Q_i, \cJ Q_j } 
    = \pair{ Q_{\frac{i+j-1}{2}}, \cJ Q_{\frac{i+j+1}{2}} }
    = \pair{ Q_{\frac{i+j-1}{2}}, \cK Q_{\frac{i+j-1}{2}} }
    = 0 . \]
    Involution with respect to $\cK$ is shown analogously.
\end{proof}

\section{Lagrangian multiforms}\label{multi-sec}

Lagrangian multiforms were proposed in \cite{lobb2009lagrangian} as a way to combine the property of multi-dimensional consistency of integrable lattice equations and their variational nature into a single framework. The continuous side of the theory was further developed in e.g.\@ \cite{suris2016variational, suris2016lagrangian, petrera2017variational, sleigh2020variational, petrera2021variational, vermeeren2021hamiltonian, caudrelier2025geometry}. Some of these works refer to this idea by the name \emph{pluri-Lagrangian systems}. In accordance to which aspect of the theory these works emphasise, we define both concepts in Definition \ref{def-multiform} below.

Consider a system of $N-1$ commuting PDEs in 2 independent variables each. We consider their combined space of independent variables $\R^N$ with coordinates $x = t_0, t_1, \ldots t_{N-1}$, where $x$ and $t_i$ are the independent variables of the $i$-th equation. The central object of Lagrangian multiform theory is a 2-form
\[ \cL = \sum_{i<j} L_{ij} \, \dd t_i \wedge \dd t_j \]
in $\R^N$, depending on the dependent variable $\bar u = (\bar u^1,\dots , \bar u^n): \R^N \to \R$ and its derivatives. For any surface $\Gamma \subset \R^N$, we define the action
\[ S_\Gamma[\bar u] = \int_\Gamma \cL[\bar u]. \]

\begin{definition}
\label{def-multiform}
    We say that $\bar u: \R^N \to \R$ satisfies the \emph{pluri-Lagrangian principle} if for every surface $\Gamma$, and every $\bar v: \R^N \to \R$ whose infinite jet prolongation vanishes on the boundary of $\Gamma$, there holds
    \[ \frac{\dd}{\dd \varepsilon} S_\Gamma[\bar u + \varepsilon \bar v] \Big|_{\varepsilon = 0} = 0 . \]
    
    We say $\bar u: \R^N \to \R$ satisfies the \emph{Lagrangian multiform principle} if it satisfies the pluri-Lagrangian principle, and in addition	
    \[ S_\Gamma[\bar u] = 0 . \]
    for all closed surfaces $\Gamma$.
\end{definition}
Note that the last condition has several equivalent versions. In particular, it is equivalent to the requirement that $\dd \cL[\bar u] = 0$.
We denote $\partial_i = \der{}{t_i}$ and $u_{t_i} = \frac{\partial u}{\partial t_i}$. More generally, we denote partial derivatives of $\bar u$ by $\bar u_I$, where $I$ is a string of $t$-variables. We denote 
\[ \bar u_{I t_i^\alpha} = \partial_i^\alpha \bar u_{I} \]
We write $I \ni t_i$ if the string $I$ contains at least one instance of $t_i$ and $I \not\ni t_i$ otherwise.
The variational derivative with respect to $\bar u_I$ in the direction of $t_i,t_j$ is defined as
    \[ \var{ij}{}{\bar u_I} := \sum_{\alpha,\beta \geq 0} (-1)^{\alpha+\beta} \partial_i^\alpha \partial_j^\beta \der{}{\bar u_{I t_i^\alpha t_j^\beta}} . \]
The variational derivative with respect to $\bar u_I$ in the direction of $t_i,t_j,t_k$ is defined as
    \[ \var{ijk}{}{\bar u_I} := \sum_{\alpha,\beta,\gamma \geq 0} (-1)^{\alpha+\beta+\gamma} \partial_i^\alpha \partial_j^\beta \partial_k^\gamma \der{}{\bar u_{I t_i^\alpha t_j^\beta t_k^\gamma}} . \]

\begin{theorem}
    \label{thm-tfae}
    Let 
    \[ \dd \cL[\bar u] = \sum_{i<j<k} P_{ijk}[\bar u] \, \dd t_i \wedge \dd t_j \wedge \dd t_k . \]
    The following are equivalent:
    \begin{enumerate}[(a)]
        \item $\bar u: \R^N \to \R$ satisfies the pluri-Lagrangian principle.
        \item $\frac{\dd}{\dd \varepsilon} \dd \cL[\bar u + \varepsilon \bar v] \Big|_{\varepsilon = 0} = 0$, i.e.\@ for all $i<j<k$, there holds $\frac{\dd}{\dd \varepsilon} P_{ijk}[\bar u + \varepsilon \bar v] \Big|_{\varepsilon = 0} = 0 .$
        \item\label{thm-partialP} For all mixed partial derivatives $\bar u_I$ of $\bar u$ and all $i<j<k$, there holds 
        \[ \der{}{\bar u_I} P_{ijk}[\bar u] = 0 . \]
        \item For all mixed partial derivatives $\bar u_I$ of $\bar u$ and all $i<j<k$, there holds 
        \[ \var{ijk}{}{\bar u_I} P_{ijk}[\bar u] = 0 . \]
        \item $\bar u: \R^N \to \R$ satisfies the system of equations
        \begin{align*}
                            &\var{ij}{L_{ij}}{\bar u} = 0 & \forall I \not\ni t_i, t_j ,\\
            &\var{ij}{L_{ij}}{\bar u_{t_i}} + \var{jk}{L_{jk}}{\bar u_{t_k}} = 0 & \forall I \not\ni t_j ,\\
            &\var{ij}{L_{ij}}{\bar u_{t_i t_j}} + \var{jk}{L_{jk}}{\bar u_{t_j t_k}} + \var{ki}{L_{ki}}{\bar u_{t_k t_i}} = 0 & \forall I ,
        \end{align*}
        where $L_{ki} = - L_{ik}$.
    \end{enumerate}
\end{theorem}

\begin{theorem}[Double zero property]
    Consider a system of equations $E[\bar u] = 0$ with
    \[ E[\bar u] = (E_1[\bar u], E_2[\bar u], \ldots, E_n[\bar u])^\top .\]
    If the coefficients of the exterior derivative of $\cL$ factorise as
    \[ P_{ijk}[\bar u] = (\cF_{ijk}[\bar u] E[\bar u]) \cdot (\cG_{ijk}[\bar u] E[\bar u]), \]
    where $\cF_{ijk}[\bar u]$ and $\cG_{ijk}[\bar u]$ are $(m \times n)$-matrix differential operators, for some $m$, then the Lagrangian multiform principle is satisfied for all functions $\bar u$ that solve the system of equations $E[\bar u] = 0$. We say that $\dd \cL$ has a \emph{double zero} at solutions of this system.
\end{theorem}
\begin{proof}
    If $\bar u$ satisfies $E[\bar u] = 0$, then it also satisfies the jet prolongations of this equations, in particular $\cF_{ijk}[\bar u] E[\bar u] = 0$ and $\cG_{ijk}[\bar u] E[\bar u] = 0$. Hence $P_{ijk}[\bar u] = 0$ and
    \begin{align*}
        \frac{\dd}{\dd \varepsilon} P_{ijk}[\bar u + \varepsilon \bar v] \Big|_{\varepsilon = 0} 
          &= \left(\frac{\dd}{\dd \varepsilon}  (\cF_{ijk}[\bar u] E[\bar u]) \Big|_{\varepsilon = 0}   \cdot (\cG_{ijk}[\bar u] E[\bar u]) \right. \\
        &\qquad \left. + (\cF_{ijk}[\bar u] E[\bar u]) \cdot \frac{\dd}{\dd \varepsilon} (\cG_{ijk}[\bar u] E[\bar u]) \Big|_{\varepsilon = 0} \right)
          =0 .
    \end{align*}
\end{proof}

\subsection{Bi-Lagrangian 2-form structure}
\label{sec-bilagrangian}

Suppose we have two compatible symplectic operators {$\cJ = \ell_p - \ell_p^*$ and $\cK = \ell_q - \ell_q^*$, such that the corresponding Hamiltonian operators $\mathcal{A}$ and $\mathcal{B}$ generate the hierarchy of Hamiltonians $h_0, h_1, \ldots$}. Assume that $\cA$ is constant, so that $\cJ$ also is.

In this subsection, we denote $x=t_0$ and $\partial_x = \partial_0$.

\subsubsection{First Lagrangian 2-form}
We will construct a Lagrangian 2-form $\cL = \sum_{i<j} L_{ij} \dd t_i \wedge \dd t_j$. First we define
\[ L_{0j} = p[\bar u] \bar u_{t_j} - h_j[\bar u] . \]
Recall that its traditional Euler-Lagrange equation is
\[ 0 = (\ell_p^* - \ell_p) \bar u_{t_j} - \var{}{h_j}{\bar u} = - \cJ  \bar u_{t_j} - \var{}{h_j}{\bar u} . \]

Let $\sim$ denote equality modulo addition of a total $x$-derivative. Then, by definition of the adjoint operator, we have, for any two functions $f,g$,
\[ (\ell_p f) g \sim f (\ell_p^* g) . \]
Aiming to have a double-zero expression for $\dd \cL$, we compute in this equivalence class:
\begin{align*}
    &\partial_k L_{0j} - \partial_j L_{0k} \\
    &= (\partial_k p) \bar u_{t_j} - (\partial_j p) \bar u_{t_k} - (\partial_k h_j) + (\partial_j h_k) \\
    &= (\ell_p \bar u_{t_k}) \bar u_{t_j} - (\ell_p \bar u_{t_j}) \bar u_{t_k} - (\partial_k h_j) + (\partial_j h_k) \\
    &\sim \frac12 (\ell_p \bar u_{t_k}) \bar u_{t_j} + \frac12 u_{t_k} (\ell_p^* \bar u_{t_j}) - \frac12(\ell_p \bar u_{t_j}) \bar u_{t_k} - \frac12 \bar u_{t_j} (\ell_p^* \bar u_{t_k}) - \bar u_{t_k} \var{}{h_j}{\bar u} + \bar u_{t_j} \var{}{h_k}{\bar u} \\
    &= \frac12 \bar u_{t_j} \cJ (\bar u_{t_k}) - \frac12 \bar u_{t_k} \cJ (\bar u_{t_j}) + \bar u_{t_k} \cJ (Q_j) - \bar u_{t_j} \cJ (Q_k) \\
    &\sim \frac12 \bar u_{t_j} \cJ (\bar u_{t_k}) - \frac12 \bar u_{t_k} \cJ (\bar u_{t_j}) + \frac12 \bar u_{t_k} \cJ (Q_j) - \frac12 Q_j \cJ (\bar u_{t_k}) - \frac12 \bar u_{t_j} \cJ (Q_k) + \frac12 Q_k \cJ (\bar u_{t_j}) \\
    &= \frac12 ( \bar u_{t_j} - Q_j) \cJ (\bar u_{t_k} - Q_k) - \frac12 ( \bar u_{t_k} - Q_k) \cJ (\bar u_{t_j} - Q_j) - \frac12 Q_j \cJ (Q_k) + \frac12 Q_k \cJ (Q_j) \\
    &\sim \frac12 ( \bar u_{t_j} - Q_j) \cJ (\bar u_{t_k} - Q_k) - \frac12 ( \bar u_{t_k} - Q_k) \cJ (\bar u_{t_j} - Q_j)
\end{align*}
The function $L_{jk}$ is defined by the computation above, by writing the terms hidden by the notation $\sim$ explicitly as $\partial_x L_{jk}$. Concretely, we take the function $L_{jk}$ such that
\[ P_{0jk} = \partial_k L_{0j} - \partial_j L_{0k} + \partial_0 L_{jk} = \frac12 ( \bar u_{t_j} - Q_j) \cJ (\bar u_{t_k} - Q_k) - \frac12 ( \bar u_{t_k} - Q_k) \cJ (\bar u_{t_j} - Q_j) ,\]
where $\partial_0 = \partial_x$.
This defines the remaining coefficients of $\cL = \sum_{i<j} L_{ij} \dd t_i \wedge \dd t_j$.

\begin{theorem}
    \label{thm-evolutionary-A}
    Assume that the coefficients of $\cL$ do not depend on $x$ explicitly (only through $\bar u$ and its derivatives), then
    $\bar u$ satisfies the Lagrangian multiform principle for $\cL$ if and only if
    \begin{equation}
        \label{multiform-evolutionary}
          \bar u_{t_j} = Q_j[\bar u]
    \end{equation}
    for all $j$.
\end{theorem}
Note that the statement of Theorem \ref{thm-evolutionary-A} does not require a bi-Hamiltonian structure. All that is assumed is a set of commuting Hamiltonian equations with respect to the same constant symplectic operator. The theorem applies regardless of whether this comes from a Lenard recursion. Its proof consists of known techniques and combines elements from \cite[Theorem 3.1]{sleigh2020variational}, \cite[Theorem 4.2]{petrera2021variational}, and \cite[Theorem 13]{sleigh2022semidiscrete}.
\begin{proof}
    By construction, the coefficients $P_{0jk}$ of $\dd \cL$ satisfy the double zero on the system of equations $ \bar u_{t_j} = Q_j[\bar u] $.  For the other coefficients of $\dd \cL$, we find
    \[ \partial_x P_{ijk} = \partial_i P_{0jk} + \partial_j P_{0ki} + \partial_k P_{0ij},
    \]
    so $\partial_x P_{ijk}$ also has the double zero property. It is easy to check that this implies condition (\ref{thm-partialP}) from Theorem \ref{thm-tfae}, i.e.\@ that
    \[ \der{P_{ijk}}{\bar u_I} = 0 \]
    on the equations \eqref{multiform-evolutionary}.
    Indeed, suppose that this is not the case and let $\bar u_I$ be the highest derivative for which 
    \[ \der{P_{ijk}}{\bar u_I} \neq 0 \]
    on the equations \eqref{multiform-evolutionary}. Since $P_{ijk}$ depends on $x$ only through $\bar u$, we have 
    \[ \partial_x P_{ijk} = \sum_{J} \der{P_{ijk}}{\bar u_J} \bar u_{Jx} . \]
    Because this quantity has the double zero property, its partial derivatives vanish on the equations \eqref{multiform-evolutionary}:
    \begin{align*}
        0 = \der{}{\bar u_{Ix}} \partial_x P_{ijk}
        &= \der{P_{ijk}}{\bar u_I} + \partial_x \der{P_{ijk}}{\bar u_{Ix}} .
    \end{align*}
    But this means that either $\der{P_{ijk}}{\bar u_I} = 0$, or $\der{P_{ijk}}{\bar u_{Ix}} \neq 0$, contradicting the assumption that $\bar u_I$ is the highest derivative for which 
    $\der{P_{ijk}}{\bar u_I} \neq 0$.
    This shows that if $\bar u$ solves the system \eqref{multiform-evolutionary}, then $\bar u$ satisfies the Lagrangian multiform principle.
    
    Now assume $\bar u$ satisfies the Lagrangian multiform principle. Let $n$ be the differential order of $\cA$ and write $\cA = a \partial_x^n + l.o.t.$, where $a$ is constant because $\cA$ is. Then by Theorem \ref{thm-tfae}(\ref{thm-partialP})
    \[ 0 = \der{P_{0jk}}{u_{(n x) t_k}} = \frac12 a (\bar u_{t_j} - Q_j) , \]
    so $\bar u$ satisfies equations \eqref{multiform-evolutionary}.
\end{proof}

\subsubsection{Second Lagrangian 2-form}
In the same way, we can construct a second Lagrangian 2-form $\cM = \sum_{i<j} M_{ij} \dd t_i \wedge \dd t_j$. First we define
\[ M_{0j} = q[\bar u] \bar u_{t_j} - h_{j+1}[\bar u] , \]
where $q$ is related to the second symplectic operator by $\cK = \ell_q - \ell_q^*$.
Analogous to the above, we find that there exists a function $M_{jk}$ such that
\[ \partial_k M_{0j} - \partial_j M_{0k} + \partial_0 M_{jk} = \frac12 ( \bar u_{t_j} - Q_j) \cK (\bar u_{t_k} - Q_k) - \frac12 ( \bar u_{t_k} - Q_k) \cK (\bar u_{t_j} - Q_j) \]
This defines the remaining coefficients of $\cM = \sum_{i<j} M_{ij} \dd t_i \wedge \dd t_j$.

\begin{theorem}
    \label{thm-evolutionary-B}
    Assume that the operator $\cB$ has constant rank. Then $\bar u$ satisfies the Lagrangian multiform principle for $\cL$ if and only if the evolutionary equations \eqref{multiform-evolutionary} hold for all $j$.
\end{theorem}
\begin{proof}
    The proof that equations \eqref{multiform-evolutionary} imply that $\bar u$ satisfies the Lagrangian multiform principle is as in Theorem \ref{thm-evolutionary-A}.

    Assume $\bar u$ satisfies the Lagrangian multiform principle. Let $n$ be the differential order of $\cB$ and write $\cB = b[\bar u] \partial_x^n + l.o.t.$, where $b[\bar u]$ is non-vanishing because $\cB$ has constant rank. Then by Theorem \ref{thm-tfae}(\ref{thm-partialP})
    \[ 0 = \der{P_{0jk}}{u_{(n x) t_k}} = \frac12 b[\bar u] (\bar u_{t_j} - Q_j) , \]
    so $\bar u$ satisfies equations \eqref{multiform-evolutionary}.
\end{proof}

\section{Examples}\label{sec4}

\subsection{Potential Korteweg-de Vries equation}
\label{sec-ex-kdv}

Let us consider the KdV equation
\[
    u_t=3uu_x+u_{3x},
\]
whose bi-Hamiltonian structure is well-established and given by the operators
$\cA = \partial_x$ and $\cB = \partial_x^3+2u\partial_x+u_x$. Let us now consider the Hamiltonian potential variable $\bar{u}$, w.r.t.\@ $\cA$, such that $u=\partial_x \bar{u}$. The equation becomes
\[
    \bar{u}_{tx}-3\bar{u}_x\bar{u}_{xx}-\bar{u}_{4x}=0,
\]
or, after integrating w.r.t.\@ $x$,
\[
    \bar{u}_t-\frac{3}{2}\bar{u}^2_x-\bar{u}_{3x}=0 ,
\]
which is known as \emph{potential KdV}, or simply pKdV. Moreover, in Hamiltonian potential variables, the second operator is mapped into
\[\bar \cB = \partial_x^3 + 2 \bar u_x \partial_x + \bar u_{xx},\]
while the first operator remains the same, $\bar \cA = \partial_x$.

We now proceed with the Lenard recursion using the symplectic operators $\bar \cA$ and $\bar \cB$.
Let $\bar u_{t_0} = \bar u_x$, then
\begin{align*}
    & \bar \cA(\bar u_{t_0}) = \bar u_{xx} = -\var{}{h_0}{\bar u}, \\
    & h_0 = \tfrac{1}{2} \bar u_{x}^{2} , \\
    & \bar \cB(\bar u_{t_0}) = 3 \bar u_{x} \bar u_{xx} + \bar u_{4x} .
\end{align*}
Solving $\bar \cB(\bar u_{t_0}) = \bar \cA(\bar u_{t_1}) = -\var{}{h_1}{\bar u}$ we find
\begin{align*}
    & \bar u_{t_1} =  \tfrac{3}{2} \bar u_{x}^{2} + \bar u_{3x} , \\
    & h_1 = \tfrac{1}{2} \bar u_{x}^{3} - \tfrac{1}{2} \bar u_{xx}^{2} , \\
    & \bar\cB(\bar u_{t_1}) = \tfrac{15}{2} \bar u_{x}^{2} \bar u_{xx} + 10 \bar u_{xx} \bar u_{3x} + 5 \bar u_{x} \bar u_{4x} + \bar u_{6x} .
\end{align*}
Solving $\bar \cB(\bar u_{t_1}) = \bar \cA(\bar u_{t_2}) = -\var{}{h_2}{\bar u}$ we find
\begin{align*}
    & \bar u_{t_2} =  \tfrac{5}{2} \bar u_{x}^{3} + \tfrac{5}{2} \bar u_{xx}^{2} + 5 \bar u_{x} \bar u_{3x} + \bar u_{5x} , \\
    & h_2 = \tfrac{5}{8} \bar u_{x}^{4} - \tfrac{5}{2} \bar u_{x} \bar u_{xx}^{2} + \tfrac{1}{2} \bar u_{3x}^{2} , \\
    &\bar\cB(\bar u_{t_2}) = \tfrac{35}{2} \bar u_{x}^{3} \bar u_{xx} + \tfrac{35}{2} \bar u_{xx}^{3} + 70 \bar u_{x} \bar u_{xx} \bar u_{3x} + \tfrac{35}{2} \bar u_{x}^{2} \bar u_{4x} + 35 \bar u_{3x} \bar u_{4x} + 21 \bar u_{xx} \bar u_{5x} + 7 \bar u_{x} \bar u_{6x} + \bar u_{8x} .
\end{align*}
Solving $\bar \cB(\bar u_{t_2}) = \bar \cA(\bar u_{t_3}) = -\var{}{h_3}{\bar u}$ we find
\begin{align*}
    & \bar u_{t_3} =  \tfrac{35}{8} \bar u_{x}^{4} + \tfrac{35}{2} \bar u_{x} \bar u_{xx}^{2} + \tfrac{35}{2} \bar u_{x}^{2} \bar u_{3x} + \tfrac{21}{2} \bar u_{3x}^{2} + 14 \bar u_{xx} \bar u_{4x} + 7 \bar u_{x} \bar u_{5x} + \bar u_{7x} , \\
    & h_3 = \tfrac{7}{8} \bar u_{x}^{5} - \tfrac{35}{4} \bar u_{x}^{2} \bar u_{xx}^{2} + \tfrac{7}{2} \bar u_{x} \bar u_{3x}^{2} - \tfrac{1}{2} \bar u_{4x}^{2} .
\end{align*}

\paragraph{First multiform}
For the symplectic operator $\bar \cA$, our construction leads to the same 2-form, up to relabelling and rescaling, as the one presented in \cite{suris2016lagrangian, vermeeren2019continuum, petrera2021variational}. Nevertheless, we find it instructive to present this example in detail.

Observe that $\bar\cA = \ell_p - \ell_p^*$ with $p = \frac12 \bar u_x$. This leads to Lagrangians of the form
\[ L = \frac12 \bar u_x \bar u_t - h, \]
which can be extended to a multiform $\cL = \sum_{i<j} L_{ij} \dd t_i \wedge \dd t_j$ with
\begin{align*}
    & L_{01} =  -\tfrac{1}{2} \bar u_{x}^{3} + \tfrac{1}{2} \bar u_{xx}^{2} + \tfrac{1}{2} \bar u_{x} \bar u_{t_1} , \\
    & L_{02} =  -\tfrac{5}{8} \bar u_{x}^{4} + \tfrac{5}{2} \bar u_{x} \bar u_{xx}^{2} - \tfrac{1}{2} \bar u_{3x}^{2} + \tfrac{1}{2} \bar u_{x} \bar u_{t_2} , \\
    & L_{03} =  -\tfrac{7}{8} \bar u_{x}^{5} + \tfrac{35}{4} \bar u_{x}^{2} \bar u_{xx}^{2} - \tfrac{7}{2} \bar u_{x} \bar u_{3x}^{2} + \tfrac{1}{2} \bar u_{4x}^{2} + \tfrac{1}{2} \bar u_{x} \bar u_{t_3} .
\end{align*}
The standard Euler-Lagrange equations of $L_{0j}$ are differential consequences of the pKdV equations, obtained by applying the operator $\bar\cA = \partial_x$ to the pKdV equations:
\begin{align*}
    & -\var{01}{L_{01}}{\bar u} = -3 \bar u_{x} \bar u_{xx} - \bar u_{4x} + \bar u_{xt_1} = \bar\cA(\bar u_{t_1} - Q_1) , \\
    & -\var{02}{L_{02}}{\bar u} = -\tfrac{15}{2} \bar u_{x}^{2} \bar u_{xx} - 10 \bar u_{xx} \bar u_{3x} - 5 \bar u_{x} \bar u_{4x} - \bar u_{6x} + \bar u_{xt_2} = \bar\cA(\bar u_{t_2} - Q_2) , \\
    & -\var{03}{L_{03}}{\bar u} = \bar \cA(\bar u_{t_3} - Q_3) .
\end{align*}
Using the construction of Section \ref{sec-bilagrangian}, we find the remaining coefficients:
\begin{align*}
    L_{12} &= \tfrac{3}{8} \bar u_{x}^{5} - \tfrac{15}{8} \bar u_{x}^{2} \bar u_{xx}^{2} + \tfrac{5}{2} \bar u_{x}^{3} \bar u_{3x} - \tfrac{5}{4} \bar u_{x}^{3} \bar u_{t_1} + \tfrac{7}{4} \bar u_{xx}^{2} \bar u_{3x} + \tfrac{3}{2} \bar u_{x} \bar u_{3x}^{2} - 3 \bar u_{x} \bar u_{xx} \bar u_{4x} + \tfrac{3}{4} \bar u_{x}^{2} \bar u_{5x} + 5 \bar u_{x} \bar u_{xx} \bar u_{xt_1}  \\
    &\qquad - \tfrac{5}{4} \bar u_{xx}^{2} \bar u_{t_1} - \tfrac{5}{2} \bar u_{x} \bar u_{3x} \bar u_{t_1} + \tfrac{3}{4} \bar u_{x}^{2} \bar u_{t_2} - \tfrac{1}{2} \bar u_{4x}^{2}  + \tfrac{1}{2} \bar u_{3x} \bar u_{5x} - \bar u_{3x} \bar u_{xxt_1} + \bar u_{4x} \bar u_{xt_1} - \bar u_{xx} \bar u_{xt_2}\\
    &\qquad - \tfrac{1}{2} \bar u_{5x} \bar u_{t_1} + \tfrac{1}{2} \bar u_{3x} \bar u_{t_2} ,
\\[1ex]
    L_{13} &= \tfrac{35}{32} \bar u_{x}^{6} - \tfrac{35}{8} \bar u_{x}^{3} \bar u_{xx}^{2} + \tfrac{175}{16} \bar u_{x}^{4} \bar u_{3x} - \tfrac{35}{16} \bar u_{x}^{4} \bar u_{t_1} + \tfrac{35}{8} \bar u_{xx}^{4} - \tfrac{35}{4} \bar u_{x} \bar u_{xx}^{2} \bar u_{3x} + \tfrac{147}{8} \bar u_{x}^{2} \bar u_{3x}^{2} - \tfrac{21}{2} \bar u_{x}^{2} \bar u_{xx} \bar u_{4x} \\
    &\qquad + \tfrac{21}{4} \bar u_{x}^{3} \bar u_{5x} + \tfrac{35}{2} \bar u_{x}^{2} \bar u_{xx} \bar u_{xt_1} - \tfrac{35}{4} \bar u_{x} \bar u_{xx}^{2} \bar u_{t_1} - \tfrac{35}{4} \bar u_{x}^{2} \bar u_{3x} \bar u_{t_1} + \tfrac{19}{4} \bar u_{3x}^{3} - 2 \bar u_{xx} \bar u_{3x} \bar u_{4x} - 5 \bar u_{x} \bar u_{4x}^{2} \\
    &\qquad + 3 \bar u_{xx}^{2} \bar u_{5x} + \tfrac{13}{2} \bar u_{x} \bar u_{3x} \bar u_{5x} - 3 \bar u_{x} \bar u_{xx} \bar u_{6x} + \tfrac{3}{4} \bar u_{x}^{2} \bar u_{7x} - 7 \bar u_{x} \bar u_{3x} \bar u_{xxt_1} + 7 \bar u_{xx} \bar u_{3x} \bar u_{xt_1} + 7 \bar u_{x} \bar u_{4x} \bar u_{xt_1} \\
    &\qquad - \tfrac{21}{4} \bar u_{3x}^{2} \bar u_{t_1} - 7 \bar u_{xx} \bar u_{4x} \bar u_{t_1} - \tfrac{7}{2} \bar u_{x} \bar u_{5x} \bar u_{t_1} + \tfrac{3}{4} \bar u_{x}^{2} \bar u_{t_3} + \tfrac{1}{2} \bar u_{5x}^{2} - \bar u_{4x} \bar u_{6x} + \tfrac{1}{2} \bar u_{3x} \bar u_{7x} + \bar u_{4x} \bar u_{3xt_1} \\
    &\qquad - \bar u_{5x} \bar u_{xxt_1} + \bar u_{6x} \bar u_{xt_1} - \bar u_{xx} \bar u_{xt_3} - \tfrac{1}{2} \bar u_{7x} \bar u_{t_1} + \tfrac{1}{2} \bar u_{3x} \bar u_{t_3} ,
\\[1ex]
    L_{23} &= \tfrac{25}{32} \bar u_{x}^{7} - \tfrac{175}{32} \bar u_{x}^{4} \bar u_{xx}^{2} + \tfrac{175}{16} \bar u_{x}^{5} \bar u_{3x} + \tfrac{175}{8} \bar u_{x} \bar u_{xx}^{4} - \tfrac{175}{8} \bar u_{x}^{2} \bar u_{xx}^{2} \bar u_{3x} + \tfrac{245}{8} \bar u_{x}^{3} \bar u_{3x}^{2} - \tfrac{35}{2} \bar u_{x}^{3} \bar u_{xx} \bar u_{4x} + \tfrac{105}{16} \bar u_{x}^{4} \bar u_{5x} \\
    &\qquad - \tfrac{35}{16} \bar u_{x}^{4} \bar u_{t_2} - \tfrac{305}{8} \bar u_{xx}^{2} \bar u_{3x}^{2} + \tfrac{95}{4} \bar u_{x} \bar u_{3x}^{3} + 20 \bar u_{xx}^{3} \bar u_{4x} - 10 \bar u_{x} \bar u_{xx} \bar u_{3x} \bar u_{4x} - \tfrac{25}{2} \bar u_{x}^{2} \bar u_{4x}^{2} + 15 \bar u_{x} \bar u_{xx}^{2} \bar u_{5x} \\
    &\qquad + \tfrac{65}{4} \bar u_{x}^{2} \bar u_{3x} \bar u_{5x} - \tfrac{15}{2} \bar u_{x}^{2} \bar u_{xx} \bar u_{6x} + \tfrac{5}{4} \bar u_{x}^{3} \bar u_{7x} + \tfrac{35}{2} \bar u_{x}^{2} \bar u_{xx} \bar u_{xt_2} - \tfrac{35}{4} \bar u_{x} \bar u_{xx}^{2} \bar u_{t_2} - \tfrac{35}{4} \bar u_{x}^{2} \bar u_{3x} \bar u_{t_2} + \tfrac{5}{4} \bar u_{x}^{3} \bar u_{t_3} \\
    &\qquad + \tfrac{19}{4} \bar u_{3x}^{2} \bar u_{5x} + 8 \bar u_{xx} \bar u_{4x} \bar u_{5x} + \tfrac{5}{2} \bar u_{x} \bar u_{5x}^{2} - 10 \bar u_{xx} \bar u_{3x} \bar u_{6x} - 5 \bar u_{x} \bar u_{4x} \bar u_{6x} + \tfrac{5}{4} \bar u_{xx}^{2} \bar u_{7x} + \tfrac{5}{2} \bar u_{x} \bar u_{3x} \bar u_{7x} \\
    &\qquad - 7 \bar u_{x} \bar u_{3x} \bar u_{xxt_2} + 7 \bar u_{xx} \bar u_{3x} \bar u_{xt_2} + 7 \bar u_{x} \bar u_{4x} \bar u_{xt_2} - 5 \bar u_{x} \bar u_{xx} \bar u_{xt_3} - \tfrac{21}{4} \bar u_{3x}^{2} \bar u_{t_2} - 7 \bar u_{xx} \bar u_{4x} \bar u_{t_2} \\
    &\qquad - \tfrac{7}{2} \bar u_{x} \bar u_{5x} \bar u_{t_2} + \tfrac{5}{4} \bar u_{xx}^{2} \bar u_{t_3} + \tfrac{5}{2} \bar u_{x} \bar u_{3x} \bar u_{t_3} - \tfrac{1}{2} \bar u_{6x}^{2} + \tfrac{1}{2} \bar u_{5x} \bar u_{7x} + \bar u_{4x} \bar u_{3xt_2} - \bar u_{5x} \bar u_{xxt_2} + \bar u_{3x} \bar u_{xxt_3} \\
    &\qquad + \bar u_{6x} \bar u_{xt_2} - \bar u_{4x} \bar u_{xt_3} - \tfrac{1}{2} \bar u_{7x} \bar u_{t_2} + \tfrac{1}{2} \bar u_{5x} \bar u_{t_3} .
\end{align*}

The multiform Euler-Lagrange equations of the form $\var{0j}{L_{0j}}{\bar u}$ give the pKdV equations in their differentiated form, for example:
\begin{align*}
    0 = \var{01}{L_{01}}{\bar u} &= 3 \bar u_{x} \bar u_{xx} + \bar u_{4x} - \bar u_{xt_1} , \\
    0 = \var{02}{L_{02}}{\bar u} &= \tfrac{15}{2} \bar u_{x}^{2} \bar u_{xx} + 10 \bar u_{xx} \bar u_{3x} + 5 \bar u_{x} \bar u_{4x} + \bar u_{6x} - \bar u_{xt_2} .
\end{align*}
The evolutionary form of the pKdV equations can be obtained in several ways from the system of multiform Euler-Lagrange equations, for example as
\begin{align*}
    0 = \var{12}{L_{12}}{\bar u_{5x}} &= \tfrac{3}{4} \bar u_{x}^{2} + \tfrac{1}{2} \bar u_{3x} - \tfrac{1}{2} \bar u_{t_1} , \\
    0 = \var{23}{L_{23}}{\bar u_{7x}} &= \tfrac{5}{4} \bar u_{x}^{3} + \tfrac{5}{4} \bar u_{xx}^{2} + \tfrac{5}{2} \bar u_{x} \bar u_{3x} + \tfrac{1}{2} \bar u_{5x} - \tfrac{1}{2} \bar u_{t_2} ,
\end{align*}
or as
\begin{align*}
    0 = \var{01}{L_{01}}{\bar u_x} + \var{12}{L_{12}}{\bar u_{t_2}} &= -\tfrac{3}{4} \bar u_{x}^{2} - \tfrac{1}{2} \bar u_{3x} + \tfrac{1}{2} \bar u_{t_1} ,\\
    0 = \var{02}{L_{02}}{\bar u_x} - \var{12}{L_{12}}{\bar u_{t_1}} &= -\tfrac{5}{4} \bar u_{x}^{3} - \tfrac{5}{4} \bar u_{xx}^{2} - \tfrac{5}{2} \bar u_{x} \bar u_{3x} - \tfrac{1}{2} \bar u_{5x} + \tfrac{1}{2} \bar u_{t_2} .
\end{align*}
By Theorem \ref{thm-evolutionary-A}, all multiform Euler-Lagrange equations are differential consequences of the system of evolutionary pKdV equations.

\paragraph{Second multiform}

In the second case, note that $\bar \cB = \ell_p - \ell_p^*$ with $p = \frac12 \bar u_{3x} + \frac12 \bar u_x^2$. This leads to Lagrangians of the form
\[ M = \tfrac12 (\bar u_{3x} + \bar u_x^2) \bar u_t - h, \]
which we can extend to a multiform $\cM = \sum_{i<j} M_{ij} \dd t_i \wedge \dd t_j$ with
\begin{align*}
& M_{01} =  -\tfrac{5}{8} \bar u_{x}^{4} + \tfrac{5}{2} \bar u_{x} \bar u_{xx}^{2} - \tfrac{1}{2} \bar u_{3x}^{2} + \tfrac{1}{2} {\left(\bar u_{x}^{2} + \bar u_{3x}\right)} \bar u_{t_1} , \\
& M_{02} =  -\tfrac{7}{8} \bar u_{x}^{5} + \tfrac{35}{4} \bar u_{x}^{2} \bar u_{xx}^{2} - \tfrac{7}{2} \bar u_{x} \bar u_{3x}^{2} + \tfrac{1}{2} \bar u_{4x}^{2} + \tfrac{1}{2} {\left(\bar u_{x}^{2} + \bar u_{3x}\right)} \bar u_{t_2} .
\end{align*}
The standard Euler-Lagrange equations of these $M_{0j}$ are obtained by applying the operator $\bar\cB$ to the pKdV equations:
\begin{align*}
-\var{01}{M_{01}}{\bar u} &= -\tfrac{15}{2} \bar u_{x}^{2} \bar u_{xx} - 10 \bar u_{xx} \bar u_{3x} - 5 \bar u_{x} \bar u_{4x} + 2 \bar u_{x} \bar u_{xt_1} + \bar u_{xx} \bar u_{t_1} - \bar u_{6x} + \bar u_{3xt_1} \\
&= \bar\cB(\bar u_{t_1} - Q_1)
\end{align*}
and
\begin{align*}
 -\var{02}{M_{02}}{\bar u} &= -\tfrac{35}{2} \bar u_{x}^{3} \bar u_{xx} - \tfrac{35}{2} \bar u_{xx}^{3} - 70 \bar u_{x} \bar u_{xx} \bar u_{3x} - \tfrac{35}{2} \bar u_{x}^{2} \bar u_{4x} - 35 \bar u_{3x} \bar u_{4x} - 21 \bar u_{xx} \bar u_{5x} - 7 \bar u_{x} \bar u_{6x} \\
 &\qquad + 2 \bar u_{x} \bar u_{xt_2} + \bar u_{xx} \bar u_{t_2} - \bar u_{8x} + \bar u_{3xt_2} \\
 &= \bar\cB(\bar u_{t_2} - Q_2) .
\end{align*}
Using the construction of Section \ref{sec-bilagrangian}, we find the remaining coefficient:
\begin{align*}
    M_{12} &= \tfrac{5}{8} \bar u_{x}^{6} + \tfrac{5}{4} \bar u_{x}^{3} \bar u_{xx}^{2} + \tfrac{55}{8} \bar u_{x}^{4} \bar u_{3x} - \tfrac{15}{8} \bar u_{x}^{4} \bar u_{t_1} + \tfrac{5}{8} \bar u_{xx}^{4} - \tfrac{25}{4} \bar u_{x} \bar u_{xx}^{2} \bar u_{3x} + \tfrac{45}{4} \bar u_{x}^{2} \bar u_{3x}^{2} + 4 \bar u_{x}^{3} \bar u_{5x} + \tfrac{5}{4} \bar u_{x}^{3} \bar u_{xxt_1} \\
    &\qquad + \tfrac{55}{4} \bar u_{x}^{2} \bar u_{xx} \bar u_{xt_1} - \tfrac{15}{2} \bar u_{x} \bar u_{xx}^{2} \bar u_{t_1} - \tfrac{35}{4} \bar u_{x}^{2} \bar u_{3x} \bar u_{t_1} + \bar u_{x}^{3} \bar u_{t_2} + 5 \bar u_{3x}^{3} + \tfrac{5}{2} \bar u_{xx} \bar u_{3x} \bar u_{4x} - \tfrac{5}{2} \bar u_{x} \bar u_{4x}^{2} + \tfrac{1}{4} \bar u_{xx}^{2} \bar u_{5x} \\
    &\qquad + \tfrac{5}{2} \bar u_{x} \bar u_{3x} \bar u_{5x} - \tfrac{3}{2} \bar u_{x} \bar u_{xx} \bar u_{6x} + \tfrac{3}{4} \bar u_{x}^{2} \bar u_{7x} + \tfrac{5}{4} \bar u_{xx}^{2} \bar u_{xxt_1} - \tfrac{9}{2} \bar u_{x} \bar u_{3x} \bar u_{xxt_1} - \tfrac{3}{4} \bar u_{x}^{2} \bar u_{xxt_2} + 2 \bar u_{xx} \bar u_{3x} \bar u_{xt_1} \\
    &\qquad + \tfrac{9}{2} \bar u_{x} \bar u_{4x} \bar u_{xt_1} - \tfrac{7}{2} \bar u_{x} \bar u_{xx} \bar u_{xt_2} - \tfrac{11}{2} \bar u_{3x}^{2} \bar u_{t_1} - \tfrac{13}{2} \bar u_{xx} \bar u_{4x} \bar u_{t_1} - \tfrac{7}{2} \bar u_{x} \bar u_{5x} \bar u_{t_1} + \bar u_{xx}^{2} \bar u_{t_2} + \tfrac{5}{2} \bar u_{x} \bar u_{3x} \bar u_{t_2} \\
    &\qquad - \tfrac{1}{2} \bar u_{4x} \bar u_{6x} + \tfrac{1}{2} \bar u_{3x} \bar u_{7x} + \bar u_{4x} \bar u_{3xt_1} - \tfrac{1}{2} \bar u_{5x} \bar u_{xxt_1} + \tfrac{1}{2} \bar u_{3x} \bar u_{xxt_2} + \tfrac{1}{2} \bar u_{6x} \bar u_{xt_1} - \tfrac{1}{2} \bar u_{4x} \bar u_{xt_2} - \tfrac{1}{2} \bar u_{7x} \bar u_{t_1} \\
    &\qquad + \tfrac{1}{2} \bar u_{5x} \bar u_{t_2} .
\end{align*}
The evolutionary equations can be found as part of the system of  multiform EL equations, for example as
\begin{align*}
    0 = \var{01}{M_{01}}{\bar u_{3x}} + \var{12}{M_{12}}{\bar u_{xxt_2}} &= -\tfrac{3}{4} \bar u_{x}^{2} - \tfrac{1}{2} \bar u_{3x} + \tfrac{1}{2} \bar u_{t_1} , \\
    0 = \var{02}{M_{02}}{\bar u_{3x}} - \var{12}{M_{12}}{\bar u_{xxt_2}} &= \tfrac{3}{4} \bar u_{x}^{2} - 7 \bar u_{x} \bar u_{3x} - \tfrac{1}{2} \bar u_{3x} - \bar u_{5x} + \tfrac{1}{2} \bar u_{t_2} .
\end{align*}
By Theorem \ref{thm-evolutionary-B}, all multiform Euler-Lagrange equations are differential consequences of the system of evolutionary pKdV equations.

\subsection{Dispersionless potential KdV}

As before, let us fix the first operator to be $\cA = \partial_x$, whose symplectic counterpart $\bar \cA = \partial_x$ corresponds to $p = \frac12 \bar u_x$, where the Hamiltonian potential variable is defined through $u = \bar u_x$. The Lagrangians associated to the operator $\cA$ are of the form $L = \frac12 \bar u_x \bar u_t - h$. In this case, we have several other candidates for the compatible operator $\cB$ \cite{olver1988hamiltonian}:
\begin{enumerate}
    \item $\cB_1 = 2 u \partial_x + u_x$
    \item $\cB_2 = u^2 \partial_x + u u_x$
    \item $\cB_3 = \partial_x \frac{1}{u_x} \partial_x\frac{1}{u_x} \partial_x$
\end{enumerate}

\paragraph{Case 1}

$\bar \cB_1 = 2 \bar u_x \partial_x + \bar u_{xx}$, so the recursion operator is $\cR = \partial_x^{-1} \circ (2 \bar u_x \partial_x + \bar u_{xx})$. More explicitly, the first few iterations of the recursion are as follows:
\begin{align*}
	& \bar u_{t_0} =  \bar u_{x} , 
	&& \bar \cA(\bar u_{t_0}) = \bar u_{xx} = -\var{}{h_0}{\bar u}, 
	&& h_0 = \tfrac{1}{2} \bar u_{x}^{2} , 
	&& \bar \cB(\bar u_{t_0}) = 3 \bar u_{x} \bar u_{xx} . \\
	& \bar u_{t_1} =  \tfrac{3}{2} \bar u_{x}^{2} , 
	&& \bar\cA(\bar u_{t_1}) = 3 \bar u_{x} \bar u_{xx} = -\var{}{h_1}{\bar u}, 
	&& h_1 = \tfrac{1}{2} \bar u_{x}^{3} , 
	&& \bar\cB(\bar u_{t_1}) = \tfrac{15}{2} \bar u_{x}^{2} \bar u_{xx} , \\
	& \bar u_{t_2} =  \tfrac{5}{2} \bar u_{x}^{3} , 
	&& \bar\cA(\bar u_{t_2}) = \tfrac{15}{2} \bar u_{x}^{2} \bar u_{xx} = -\var{}{h_2}{\bar u}, 
	&& h_2 = \tfrac{5}{8} \bar u_{x}^{4} , 
	&&\bar \cB(\bar u_{t_2}) = \tfrac{35}{2} \bar u_{x}^{3} \bar u_{xx} , \\
	& \bar u_{t_3} =  \tfrac{35}{8} \bar u_{x}^{4} , 
	&&\bar \cA(\bar u_{t_3}) = \tfrac{35}{2} \bar u_{x}^{3} \bar u_{xx} = -\var{}{h_3}{\bar u},
	&& h_3 = \tfrac{7}{8} \bar u_{x}^{5} ,
	&& \ldots
\end{align*}
The Lagrangian multiform associated to $\cA$ has coefficients
{\begin{align*}
 & L_{01} =  -\tfrac{1}{2} \bar u_{x}^{3} + \tfrac{1}{2} \bar u_{x} \bar u_{t_1} ,
	&& L_{02} =  -\tfrac{5}{8} \bar u_{x}^{4} + \tfrac{1}{2} \bar u_{x} \bar u_{t_2} ,
	&& L_{03} =  -\tfrac{7}{8} \bar u_{x}^{5} + \tfrac{1}{2} \bar u_{x} \bar u_{t_3} ,
\end{align*}}
and
{\begin{align*}
	& L_{12} = \tfrac{3}{8} \bar u_{x}^{5} - \tfrac{5}{4} \bar u_{x}^{3} \bar u_{t_1} + \tfrac{3}{4} \bar u_{x}^{2} \bar u_{t_2} ,
	&& L_{13} = \tfrac{35}{32} \bar u_{x}^{6} - \tfrac{35}{16} \bar u_{x}^{4} \bar u_{t_1} + \tfrac{3}{4} \bar u_{x}^{2} \bar u_{t_3} ,
	&& L_{23} = \tfrac{25}{32} \bar u_{x}^{7} - \tfrac{35}{16} \bar u_{x}^{4} \bar u_{t_2} + \tfrac{5}{4} \bar u_{x}^{3} \bar u_{t_3} .
\end{align*}}
It can be obtained from the KdV multiform by removing all terms depending on second and higher derivatives.

On the other hand, we can write $\bar\cB_{1} = \ell_{p_1} -  \ell^*_{p_1}$, where $p_1 = \frac12 \bar u_x^2$. Hence the Lagrangian associated to $\cB_1$ is
\[ M = \tfrac12 \bar u_x^2 \bar u_t - h . \]
and the Lagrangian multiform associated to $\bar\cB_1$ has coefficients
\begin{align*}
	& M_{01} =  -\tfrac{5}{8} \bar u_{x}^{4} + \tfrac{1}{2} \bar u_{x}^{2} \bar u_{t_1},
	&& M_{02} =  -\tfrac{7}{8} \bar u_{x}^{5} + \tfrac{1}{2} \bar u_{x}^{2} \bar u_{t_2},
\end{align*}
and
\[ M_{12} = \tfrac{5}{8} \bar u_{x}^{6} - \tfrac{15}{8} \bar u_{x}^{4} \bar u_{t_1} + \bar u_{x}^{3} \bar u_{t_2} . \]

\paragraph{Case 2}

$\bar \cB_2 =  \bar u_x^2 \partial_x + \bar u_x \bar u_{xx}$, so the recursion operator is $\cR = \partial_x^{-1} \circ (\bar u_x^2 \partial_x + \bar u_x \bar u_{xx}) $. The relevant quantities at the first few levels of the recursion are
\begin{align*}
	& \bar u_{t_0} =  \bar u_{x} , 
	&& \bar\cA(\bar u_{t_0}) = \bar u_{xx} = -\var{}{h_0}{\bar u}, 
	&& h_0 = \tfrac{1}{2} \bar u_{x}^{2} ,
	&& \bar\cB(\bar u_{t_0}) = 2 \bar u_{x}^{2} \bar u_{xx} ,  \\
	& \bar u_{t_1} =  \tfrac{2}{3} \bar u_{x}^{3} , 
	&& \bar\cA(\bar u_{t_1}) = 2 \bar u_{x}^{2} \bar u_{xx} = -\var{}{h_1}{\bar u}, 
	&& h_1 = \tfrac{1}{6} \bar u_{x}^{4} ,
	&& \bar\cB(\bar u_{t_1}) = \tfrac{8}{3} \bar u_{x}^{4} \bar u_{xx} , \\
	& \bar u_{t_2} =  \tfrac{8}{15} \bar u_{x}^{5} , 
	&& \bar\cA(\bar u_{t_2}) = \tfrac{8}{3} \bar u_{x}^{4} \bar u_{xx} = -\var{}{h_2}{\bar u}, 
	&& h_2 = \tfrac{4}{45} \bar u_{x}^{6} , 
	&& \bar\cB(\bar u_{t_2}) = \tfrac{16}{5} \bar u_{x}^{6} \bar u_{xx} , \\
	& \bar u_{t_3} =  \tfrac{16}{35} \bar u_{x}^{7} ,
	&& \bar\cA(\bar u_{t_3}) = \tfrac{16}{5} \bar u_{x}^{6} \bar u_{xx} = -\var{}{h_3}{\bar u}, 
	&& h_3 = \tfrac{2}{35} \bar u_{x}^{8} ,
	&& \ldots
\end{align*}
Up to rescaling, the $i$-th level of this hierarchy is the $2i$-th level of the previous example.

Here, we can write $\bar\cB_2 = \ell_{p_2} -  \ell^*_{p_2}$, where $p_2 = \frac16 \bar u_x^3$. Hence the Lagrangian associated to $\cB_2$ is
\[ M = \tfrac16 \bar u_x^3 \bar u_t - h . \]
The corresponding multiform has coefficients
\begin{align*}
	& M_{01} =  -\tfrac{4}{45} \bar u_{x}^{6} + \tfrac{1}{6} \bar u_{x}^{3} \bar u_{t_1},
	&& M_{02} =  -\tfrac{2}{35} \bar u_{x}^{8} + \tfrac{1}{6} \bar u_{x}^{3} \bar u_{t_2},
\end{align*}
and
\[M_{12} = \tfrac{8}{225} \bar u_{x}^{10} - \tfrac{4}{21} \bar u_{x}^{7} \bar u_{t_1} + \tfrac{1}{5} \bar u_{x}^{5} \bar u_{t_2}. \]

\paragraph{Case 3}

$\bar \cB_3 = \partial_x \frac{1}{\bar u_{xx}} \partial_x \frac{1}{\bar u_{xx}} \partial_x$, so the recursion operator is $\cR =  \frac{1}{\bar u_{xx}} \partial_x \frac{1}{\bar u_{xx}} \partial_x$.
It is not clear where we should start the Lenard recursion, because $\bar\cB_3 (\bar u_x) = 0$. However, with the Hamiltonians 
\[ h_1 = 3 \bar u_x^2, \qquad h_2 = 4 \bar u_x^3, \]
we have 
\[ \bar\cB_3(\bar u_x^3) = -\var{}{h_1}{\bar u}, \qquad \bar \cB_3(\bar u_x^4) = -\var{}{h_2}{\bar u} . \]

We can write $\bar \cB_3 = \ell_{p_3} -  \ell^*_{p_3}$, where $p_3 = \frac12 \frac{\bar u_{xxx}}{\bar u_{xx}^2}$. Hence the Lagrangian associated to $\cB_3$ is
\[ M = \frac12 \frac{\bar u_{xxx}}{\bar u_{xx}^2} \bar u_t - h . \]
(An equivalent Lagrangian is $M = \frac12 \frac{1}{\bar u_{xx}} \bar u_{xt} - h$.)
Then the Lagrangian multiform associated the Hamiltonians above, for the equations
\[ \bar u_{t_1} = \bar u_x^3, \qquad \bar u_{t_2} = \bar u_x^4,\]
is
\begin{align*}
	M_{01} &=  -3 \bar u_{x}^{2} + \frac{\bar u_{xxx} \bar u_{t_1}}{2 \bar u_{xx}^{2}} , \\
	M_{02} &=  -4 \bar u_{x}^{3} + \frac{\bar u_{xxx} \bar u_{t_2}}{2 \bar u_{xx}^{2}} , \\
	M_{12} &= \tfrac{9}{5} \bar u_{x}^{5} + \frac{\bar u_{x}^{4} \bar u_{xxt_1}}{2 \bar u_{xx}^{2}} - \frac{2 \bar u_{x}^{3} \bar u_{xt_1}}{\bar u_{xx}} - \frac{\bar u_{x}^{4} \bar u_{xxx} \bar u_{xt_1}}{2 \bar u_{xx}^{3}} - 6 \bar u_{x}^{2} \bar u_{t_1} - \frac{\bar u_{x}^{3} \bar u_{xxt_2}}{2 \bar u_{xx}^{2}} + \frac{3 \bar u_{x}^{2} \bar u_{xt_2}}{2 \bar u_{xx}} + \frac{\bar u_{x}^{3} \bar u_{xxx} \bar u_{xt_2}}{2 \bar u_{xx}^{3}} \\
	&\qquad + 3 \bar u_{x} \bar u_{t_2} - \frac{\bar u_{xxx} \bar u_{xt_2} \bar u_{t_1}}{2 \bar u_{xx}^{3}} + \frac{\bar u_{xxx} \bar u_{xt_1} \bar u_{t_2}}{2 \bar u_{xx}^{3}} .
\end{align*}
Note that in this example, $M_{12}$ contains terms depending on both a $t_1$-derivative and a $t_2$-derivative.

The evolutionary equations can be obtained from the system of multiform Euler-Lagrange equations as
\begin{align*}
    0 = \var{01}{M_{01}}{\bar u_{3x}} + \var{12}{M_{12}}{\bar u_{xxt_2}} &= -\frac{\bar u_{x}^{3}}{2 \bar u_{xx}^{2}} + \frac{\bar u_{t_1}}{2 \bar u_{xx}^{2}} ,\\
    0 = \var{02}{M_{02}}{\bar u_{3x}} - \var{12}{M_{12}}{\bar u_{xxt_1}} &= -\frac{\bar u_{x}^{4}}{2 \bar u_{xx}^{2}} + \frac{\bar u_{t_2}}{2 \bar u_{xx}^{2}} .
\end{align*}

\subsection{Polytropic gas dynamics}
Our construction also applies to multi-component evolutionary equations. The corresponding Hamiltonian structures are here described by matrix differential operators. In this subsection, we apply our procedure to a two-component Lagrangian equation arising in the context of polytropic gas dynamics (we refer to  \cite{nutku1987new, olver1988hamiltonian} and \cite[Section V]{nutku2002multilagrangians}). Consider a two-component system with variables $u,v$ and Hamiltonian operator $$\cA = \begin{pmatrix} 0 & \partial_x \\ \partial_x & 0 \end{pmatrix}.$$ This means we take potential variables $\bar u, \bar v$ that satisfy $\bar u_x = v$ and $\bar v_x = u$. 
The operator $\cA$ can be written as $\cA = \ell_{p_0} - \ell_{p_0}^*$, with $p_0 =\begin{pmatrix} 0 & \bar u_x \end{pmatrix}^\top$.
Hence, the Lagrangian associated to the operator $\cA$ is of the form\footnote{Equivalently, we could take $p_0 =\begin{pmatrix} \bar v_x & 0 \end{pmatrix}^\top$ and $L = \bar v_x \bar u_t - h$, or a linear combination of the two.}
\[ L = \bar u_x \bar v_t - h . \]
A compatible Hamiltonian operator is given by
\begin{align*}
	\cB &= \begin{pmatrix} v^{\gamma-2} \partial_x + \partial_x v^{\gamma-2} & (\gamma-1) u \partial_x + u_x \\ (\gamma-1) u \partial_x + (\gamma-2) u_x & v \partial_x + \partial_x v \end{pmatrix} ,
\end{align*}
where $\gamma$ is an arbitrary parameter. In the Hamiltonian potential variables for $\mathcal{A}$, we have
	\begin{align*} \bar{\mathcal{B}}&= \begin{pmatrix} \bar u_x^{\gamma-2} \partial_x + \partial_x \bar u_x^{\gamma-2} & (\gamma-1) \bar v_x \partial_x + \bar v_{xx} \\ (\gamma-1) \bar v_x \partial_x + (\gamma-2) \bar v_{xx} & \bar u_x \partial_x + \partial_x \bar u_x \end{pmatrix},
\end{align*}
which we can write as $\cB = \ell_{p_1} -  \ell^*_{p_1}$, where 
\[ p_1 = \begin{pmatrix} \frac{1}{\gamma-1} \bar u_x^{\gamma-1} + \frac{\gamma-2}{2} \bar v_x^2 \\  \bar u_x \bar v_x \end{pmatrix}. \]
Hence the Lagrangian associated to $\cB$ is of the form
\[ M = \frac{1}{\gamma-1} \bar u_x^{\gamma-1} \bar u_t + \frac{\gamma-2}{2} \bar v_x^2 \bar u_t + \bar u_x \bar v_x \bar v_t - h\]

We focus on the special case $\gamma = 2$, for which 
\begin{align*}
	\bar\cB
	&= \begin{pmatrix}  2 \partial_x & \partial_x \bar v_x  \\ \bar v_x \partial_x & \bar u_x \partial_x + \partial_x \bar u_x \end{pmatrix}
\end{align*}
Denote by $\delta h$ the vector of variational derivatives 
\[\delta h = \begin{pmatrix}
    \displaystyle \var{}{h}{\bar u} & \displaystyle \var{}{h}{\bar v}
\end{pmatrix}^\top . \]
Then the first few steps of the Lenard recursion are given by
\begin{align*}
& \begin{cases} \bar u_{t_0} = \bar u_{x}, \\ \bar v_{t_0} = \bar v_{x}, \end{cases}
&& \bar\cA\begin{pmatrix}\bar u_{t_0}\\\bar v_{t_0}\end{pmatrix} = \begin{pmatrix}\bar v_{xx} \\ \bar u_{xx}\end{pmatrix} = -\delta h_0, \\
& h_0 = \bar u_{x} \bar v_{x} ,
&& \bar\cB\begin{pmatrix}\bar u_{t_0}\\\bar v_{t_0}\end{pmatrix} = \begin{pmatrix}2 \bar v_{x} \bar v_{xx} + 2 \bar u_{xx} \\ 2 \bar u_{xx} \bar v_{x} + 2 \bar u_{x} \bar v_{xx}\end{pmatrix} ,
\\[1ex]
& \begin{cases} \bar u_{t_1} =  2 \bar u_{x} \bar v_{x}, \\  \bar v_{t_1} = \bar v_{x}^{2} + 2 \bar u_{x}, \end{cases}
&& \bar\cA\begin{pmatrix}\bar u_{t_1}\\\bar v_{t_1}\end{pmatrix} = \begin{pmatrix}2 \bar v_{x} \bar v_{xx} + 2 \bar u_{xx} \\ 2 \bar u_{xx} \bar v_{x} + 2 \bar u_{x} \bar v_{xx}\end{pmatrix} = -\delta h_1, \\
& h_1 = \bar u_{x} \bar v_{x}^{2} + \bar u_{x}^{2} , 
&& \bar\cB\begin{pmatrix}\bar u_{t_1}\\\bar v_{t_1}\end{pmatrix} = \begin{pmatrix}3 \bar v_{x}^{2} \bar v_{xx} + 6 \bar u_{xx} \bar v_{x} + 6 \bar u_{x} \bar v_{xx} \\ 3 \bar u_{xx} \bar v_{x}^{2} + 6 \bar u_{x} \bar v_{x} \bar v_{xx} + 6 \bar u_{x} \bar u_{xx}\end{pmatrix} ,
\\[1ex]
& \begin{cases} \bar u_{t_2} = 3 \bar u_{x} \bar v_{x}^{2} + 3 \bar u_{x}^{2}, \\ \bar v_{t_2} = \bar v_{x}^{3} + 6 \bar u_{x} \bar v_{x}, \end{cases}
&&\bar \cA\begin{pmatrix}\bar u_{t_2}\\\bar v_{t_2}\end{pmatrix} = \begin{pmatrix}3 \bar v_{x}^{2} \bar v_{xx} + 6 \bar u_{xx} \bar v_{x} + 6 \bar u_{x} \bar v_{xx} \\ 3 \bar u_{xx} \bar v_{x}^{2} + 6 \bar u_{x} \bar v_{x} \bar v_{xx} + 6 \bar u_{x} \bar u_{xx}\end{pmatrix} = -\delta h_2, \\
& h_2 = \bar u_{x} \bar v_{x}^{3} + 3 \bar u_{x}^{2} \bar v_{x} , 
&& \bar\cB\begin{pmatrix}\bar u_{t_2}\\\bar v_{t_2}\end{pmatrix} = \begin{pmatrix}4 \bar v_{x}^{3} \bar v_{xx} + 12 \bar u_{xx} \bar v_{x}^{2} + 24 \bar u_{x} \bar v_{x} \bar v_{xx} + 12 \bar u_{x} \bar u_{xx} \\ 4 \bar u_{xx} \bar v_{x}^{3} + 12 \bar u_{x} \bar v_{x}^{2} \bar v_{xx} + 24 \bar u_{x} \bar u_{xx} \bar v_{x} + 12 \bar u_{x}^{2} \bar v_{xx}\end{pmatrix} ,
\\[1ex]
& \begin{cases} \bar u_{t_3} = 4 \bar u_{x} \bar v_{x}^{3} + 12 \bar u_{x}^{2} \bar v_{x}, \\ \bar v_{t_3} = \bar v_{x}^{4} + 12 \bar u_{x} \bar v_{x}^{2} + 6 \bar u_{x}^{2}, \end{cases}
&&\bar \cA\begin{pmatrix}\bar u_{t_3}\\\bar v_{t_3}\end{pmatrix} = \begin{pmatrix}4 \bar v_{x}^{3} \bar v_{xx} + 12 \bar u_{xx} \bar v_{x}^{2} + 24 \bar u_{x} \bar v_{x} \bar v_{xx} + 12 \bar u_{x} \bar u_{xx} \\ 4 \bar u_{xx} \bar v_{x}^{3} + 12 \bar u_{x} \bar v_{x}^{2} \bar v_{xx} + 24 \bar u_{x} \bar u_{xx} \bar v_{x} + 12 \bar u_{x}^{2} \bar v_{xx}\end{pmatrix}, \\
& h_3 = \bar u_{x} \bar v_{x}^{4} + 6 \bar u_{x}^{2} \bar v_{x}^{2} + 2 \bar u_{x}^{3} , 
&& \ldots
\end{align*}
The Lagrangian multiform associated to the operator $\bar\cA$ has coefficients
\begin{align*}
	& L_{01} = \bar u_{t_1} \bar v_{x} -\bar u_{x} \bar v_{x}^{2} - \bar u_{x}^{2} , \\
	& L_{02} = \bar u_{t_2} \bar v_{x} -\bar u_{x} \bar v_{x}^{3} - 3 \bar u_{x}^{2} \bar v_{x} , \\
	& L_{03} = \bar u_{t_3} \bar v_{x} -\bar u_{x} \bar v_{x}^{4} - 6 \bar u_{x}^{2} \bar v_{x}^{2} - 2 \bar u_{x}^{3} ,
\end{align*}
and
\begin{align*}
	L_{12} &= \tfrac{1}{2} \bar u_{x} \bar v_{x}^{4} + \tfrac{3}{2} \bar u_{x}^{2} \bar v_{x}^{2} - \tfrac{1}{2} \bar u_{t_1} \bar v_{x}^{3} - \tfrac{3}{2} \bar u_{x} \bar v_{x}^{2} \bar v_{t_1} + \bar u_{x}^{3} - 3 \bar u_{x} \bar u_{t_1} \bar v_{x} + \tfrac{1}{2} \bar u_{t_2} \bar v_{x}^{2} \\
	&\qquad - \tfrac{3}{2} \bar u_{x}^{2} \bar v_{t_1} + \bar u_{x} \bar v_{x} \bar v_{t_2} + \bar u_{x} \bar u_{t_2} + \tfrac{1}{2} \bar u_{t_2} \bar v_{t_1} - \tfrac{1}{2} \bar u_{t_1} \bar v_{t_2} , \\
	L_{13} &= \bar u_{x} \bar v_{x}^{5} + 6 \bar u_{x}^{2} \bar v_{x}^{3} - \tfrac{1}{2} \bar u_{t_1} \bar v_{x}^{4} - 2 \bar u_{x} \bar v_{x}^{3} \bar v_{t_1} + 6 \bar u_{x}^{3} \bar v_{x} - 6 \bar u_{x} \bar u_{t_1} \bar v_{x}^{2} - 6 \bar u_{x}^{2} \bar v_{x} \bar v_{t_1} - 3 \bar u_{x}^{2} \bar u_{t_1} \\
	&\qquad + \tfrac{1}{2} \bar u_{t_3} \bar v_{x}^{2} + \bar u_{x} \bar v_{x} \bar v_{t_3} + \bar u_{x} \bar u_{t_3} + \tfrac{1}{2} \bar u_{t_3} \bar v_{t_1} - \tfrac{1}{2} \bar u_{t_1} \bar v_{t_3} , \\
	L_{23} &= \tfrac{1}{2} \bar u_{x} \bar v_{x}^{6} + \tfrac{9}{2} \bar u_{x}^{2} \bar v_{x}^{4} + 9 \bar u_{x}^{3} \bar v_{x}^{2} - \tfrac{1}{2} \bar u_{t_2} \bar v_{x}^{4} - 2 \bar u_{x} \bar v_{x}^{3} \bar v_{t_2} - 6 \bar u_{x} \bar u_{t_2} \bar v_{x}^{2} + \tfrac{1}{2} \bar u_{t_3} \bar v_{x}^{3} \\
	&\qquad - 6 \bar u_{x}^{2} \bar v_{x} \bar v_{t_2} + \tfrac{3}{2} \bar u_{x} \bar v_{x}^{2} \bar v_{t_3} - 3 \bar u_{x}^{2} \bar u_{t_2} + 3 \bar u_{x} \bar u_{t_3} \bar v_{x} + \tfrac{3}{2} \bar u_{x}^{2} \bar v_{t_3} + \tfrac{1}{2} \bar u_{t_3} \bar v_{t_2} - \tfrac{1}{2} \bar u_{t_2} \bar v_{t_3} .
\end{align*}
The Lagrangian multiform associated to the operator $\bar\cB$ has coefficients
\begin{align*}
	M_{01} &= \bar u_{x} \bar u_{t_1} -\bar u_{x} \bar v_{x}^{3} - 3 \bar u_{x}^{2} \bar v_{x} + \bar u_{x} \bar v_{x} \bar v_{t_1} , \\
	M_{02} &= \bar u_{x} \bar u_{t_2} -\bar u_{x} \bar v_{x}^{4} - 6 \bar u_{x}^{2} \bar v_{x}^{2} - 2 \bar u_{x}^{3} + \bar u_{x} \bar v_{x} \bar v_{t_2} ,
\end{align*}
and
\begin{align*}
	M_{12} &= \tfrac{1}{2} \bar u_{x} \bar v_{x}^{5} + \tfrac{7}{2} \bar u_{x}^{2} \bar v_{x}^{3} - \tfrac{1}{2} \bar u_{t_1} \bar v_{x}^{4} - \tfrac{3}{2} \bar u_{x} \bar v_{x}^{3} \bar v_{t_1} + 3 \bar u_{x}^{3} \bar v_{x} - 6 \bar u_{x} \bar u_{t_1} \bar v_{x}^{2} + \tfrac{1}{2} \bar u_{t_2} \bar v_{x}^{3} - \tfrac{9}{2} \bar u_{x}^{2} \bar v_{x} \bar v_{t_1} \\
	&\qquad + \bar u_{x} \bar v_{x}^{2} \bar v_{t_2} - 3 \bar u_{x}^{2} \bar u_{t_1} + 3 \bar u_{x} \bar u_{t_2} \bar v_{x} - \tfrac{1}{2} \bar u_{t_2} \bar v_{x} \bar v_{t_1} + \bar u_{x}^{2} \bar v_{t_2} + \tfrac{1}{2} \bar u_{t_1} \bar v_{x} \bar v_{t_2} .
\end{align*}

\subsection{The constant astigmatism equation}

In this subsection, we apply our results for a Lagrangian equation whose bi-Hamiltonian structure has been recently investigated and, as far as the authors know, never described in terms of symplectic operators. This example comes out from the classical theory of surfaces.  The study of surfaces immersed in the Euclidean space and having constant and non-zero difference between the principal radii of curvatures dates back to the 19th century (see e.g. \cite{bianchi1879ricerchea,lipschitz1900zur}). These surfaces are known as surfaces of \emph{constant astigmatism} and are parametrised by the equation \begin{equation}\label{constast}
    u_{tt}+\left(\frac{1}{u}\right)_{xx}+2=0,
\end{equation}
which is called the \emph{constant astigmatism equation}, or CAE. Here $u=u(t,x)$ is the single function  expressing all the nonzero
coefficients of the three fundamental forms of a surface of constant
astigmatism.

In recent years, equation \eqref{constast} was revisited in \cite{baran2009integrability} and its solutions have been deeply investigated in \cite{manganaro2014constant, hlavac2015multisoliton,hlavac2017nonlocal,hlavac2018more}. An interesting result also relates the CAE with the sine-Gordon equation by a reciprocal transformation \cite{hlavac2014reciprocal}.

In \cite{pavlov2013lagrangian}, the authors introduced a new variable $w$, such that $w_x=u_t$, and mapped the CAE into a system of non-homogeneous quasilinear evolutionary PDEs: 
\begin{equation}\label{const-sys}
u_t=w_x,\qquad w_t+\left(\frac{1}{u}\right)_x+2x=0.
\end{equation}
In the same paper, they proved that the system is bi-Hamiltonian, with
\begin{align*}
    &\cA = \begin{pmatrix}
        0&1\\1&0
    \end{pmatrix}\partial_x, 
    \qquad H_1=\int{\left(\frac{1}{2}w^2-\log{u}-x^2u\right)\, dx},\\[4pt]
    & \cB = \begin{pmatrix}2u\partial_x+u_x &-w_x\\ w_x&\frac{2}{u}\partial_x-\frac{u_x}{u^2}+2\partial_x^{-1}\end{pmatrix}, \qquad 
    H_2=-\int{w\, dx}.
\end{align*}
The corresponding recursion operator is given by 
\[
    \cR=\cB \circ \cA^{-1}=  
    \begin{pmatrix}
        -w_x&2u\partial_x+u_x\\\frac{2}{u}\partial_x-\frac{u_x}{u^2}+2\partial_x^{-1}&w_x
    \end{pmatrix} \partial_x^{-1},
\]

We remark that the first Hamiltonian structure is of Dubrovin-Novikov type and in flat coordinates, whereas the second is of Ferapontov type (see \cite{ferapontov1992nonlocal}), i.e.\@ it is the sum of a first-order homogeneous operator plus a nonlocal term
\[
    \cB^{ij}=g^{ij}\, \partial_x-g^{is}\Gamma^j_{sk}\, u^j_x + f^i\partial_x^{-1}f^j, 
\]
with $g^{ij}=2u\, \partial_u\otimes \partial_u+\frac{2}{u}\partial_w\otimes \partial_w$ and $f=\sqrt{2}\, \partial_w$. One can prove that these operators are generated by nonlocal extensions of Dubrovin-Novikov operators with infinitesimal isometries of the leading coefficient $g$.    Pairs of compatible operators of this type have been investigated and classified in \cite{pavlov2021classification}.

\paragraph{Bi-symplectic and bi-Lagrangian formalism} 

Let us now consider the potential variables with respect to the first Hamiltonian structure, i.e. $\bar{u},\bar{w}$ so that 
\[
    u=\bar{w}_x, \quad w=\bar{u}_x.
\]
Taking $p_1= \begin{pmatrix} \bar{w}_x & 0\end{pmatrix}^\top$, we have 
\[
    \cA = \ell_{p_1}-\ell^*_{p_1}
    = \begin{pmatrix}
        0&1\\1&0
    \end{pmatrix}\partial_x.
\]
Hence the Lagrangians associated to $\cA$ are of the form
\[ L = \bar w_x \bar u_t - h .\]

The second Hamiltonian operator can be written as $\cB = \ell_p - \ell_p^*$ with
\[ p = \begin{pmatrix}  -\bar u_{xx} \bar w & \log(\bar w_x) + \partial_x^{-1} \bar w \end{pmatrix}^\top, \]
where sum in the definitions \eqref{frechet}--\eqref{frechet-adjoint} of $\ell$ and $\ell^*$ are extended to allow negative $k$. So, the Lagrangians associated to $\cB$ are of the form
\[ M = -\bar u_{xx} \bar w \bar u_t + \log(\bar w_x) \bar w_t + (\partial_x^{-1} \bar w) \bar w_t - h .\]

Despite the non-locality, we can find a second equation from the symplectic Lenard recursion:
\begin{align*}
    & \begin{cases} \bar u_{t_1} = -x^{2} - \frac{1}{\bar w_{x}}, \\ \bar w_{t_1} = \bar u_{x}, \end{cases}
    && \cA\begin{pmatrix}\bar u_{t_1}\\\bar w_{t_1}\end{pmatrix} = \begin{pmatrix}\bar u_{xx} \\ -2 x + \frac{\bar w_{xx}}{\bar w_{x}^{2}}\end{pmatrix} = - \delta h_1, \\
    & h_1 = \tfrac{1}{2} \bar u_{x}^{2} + 2 x \bar w - \log(\bar w_x) ,
    && \cB\begin{pmatrix}\bar u_{t_1}\\\bar w_{t_1}\end{pmatrix} = \begin{pmatrix}-x^{2} \bar w_{xx} - \bar u_{x} \bar u_{xx} - 4 x \bar w_{x} + \frac{\bar w_{xx}}{\bar w_{x}} \\ -x^{2} \bar u_{xx} + \frac{\bar u_{xx}}{\bar w_{x}} - \frac{\bar u_{x} \bar w_{xx}}{\bar w_{x}^{2}} + 2 \bar u\end{pmatrix} ,
    \\[1ex]
    & \begin{cases} \bar u_{t_2} = -x^{2} \bar u_{x} + 2 x \bar u + \frac{\bar u_{x}}{\bar w_{x}}, \\ \bar w_{t_2} =   -x^{2} \bar w_{x} - \frac{1}{2} \bar u_{x}^{2} - 2 x \bar w + \log(\bar w_{x}) + 2 (\partial_x^{-1} \bar w), \end{cases}
    && \cA\begin{pmatrix}\bar u_{t_2}\\\bar w_{t_2}\end{pmatrix} = \begin{pmatrix}-x^{2} \bar w_{xx} - \bar u_{x} \bar u_{xx} - 4 x \bar w_{x} + \frac{\bar w_{xx}}{\bar w_{x}} \\ -x^{2} \bar u_{xx} + \frac{\bar u_{xx}}{\bar w_{x}} - \frac{\bar u_{x} \bar w_{xx}}{\bar w_{x}^{2}} + 2 \bar u\end{pmatrix} ,\\
    & h_2 = x^{2} \bar u_{xx} \bar w - \tfrac{1}{6} \bar u_{x}^{3} + \log(\bar w_{x}) \bar u_{x} - 2 \bar u \bar w , 
    && \ldots
\end{align*}
We see that the second flow involves the non-local variable $\partial_x^{-1} \bar w$. Continuing the recursion would require the introduction of additional non-local variables. For example, the expression we get by applying $\cB$ to $\begin{pmatrix}\bar u_{t_2} & \bar w_{t_2}\end{pmatrix}^\top$ involves $\partial_x^{-1}\left(- \tfrac{1}{2} \bar u_{x}^{2} + \log(\bar w_{x}) + 2 (\partial_x^{-1} \bar w) \right)$.
It is not clear if a local expression of $h_3$ exists.

A Lagrangian multiform based on the operator $\cA$ has coefficients
\begin{align*}
    & L_{01} = \bar u_{t_1} \bar w_{x} -\tfrac{1}{2} \bar u_{x}^{2} - 2 x \bar w + \log\left(\bar w_{x}\right) , \\
    & L_{02} = \bar u_{t_2} \bar w_{x} -x^{2} \bar u_{xx} \bar w + \tfrac{1}{6} \bar u_{x}^{3} - \log\left(\bar w_{x}\right) \bar u_{x} + 2 \bar u \bar w , \\
    & L_{12} = \tfrac{1}{2} x^{4} \bar w_{x} + \tfrac{1}{4} x^{2} \bar u_{x}^{2} - x^{3} \bar w - x^{2} \bar u_{xt_1} \bar w + \tfrac{1}{2} x^{2} \bar u_{t_1} \bar w_{x} - \tfrac{1}{2} x^{2} \bar u_{x} \bar w_{t_1} + \tfrac{1}{2} x^{2} \log\left(\bar w_{x}\right) - x \bar u \bar u_{x} + \tfrac{1}{4} \bar u_{x}^{2} \bar u_{t_1} \\
    &\qquad + x \bar u_{t_1} \bar w + x \bar u \bar w_{t_1} + \tfrac{1}{2} x^{2} \bar w_{t_2} + x^{2} (\partial_x^{-1} \bar w) + \tfrac{3}{2} x^{2} + \bar u^{2} - \tfrac{1}{2} \log\left(\bar w_{x}\right) \bar u_{t_1} + \tfrac{1}{2} \bar u_{x} \bar u_{t_2} + \tfrac{1}{2} \bar u_{t_2} \bar w_{t_1} - \tfrac{1}{2} \bar u_{t_1} \bar w_{t_2} \\
    &\qquad + \bar u_{t_1} (\partial_x^{-1} \bar w) + \frac{\bar u_{x}^{2}}{4 \bar w_{x}} - \frac{x \bar w}{\bar w_{x}} - \frac{\bar u_{x} \bar w_{t_1}}{2 \bar w_{x}} + \frac{\log\left(\bar w_{x}\right)}{2 \bar w_{x}} - \frac{\bar w_{t_2}}{2 \bar w_{x}} + \frac{(\partial_x^{-1} \bar w)}{\bar w_{x}} + \frac{1}{\bar w_{x}} .
\end{align*}
We can check that the evolution equations are part of the system of multiform Euler-Lagrange equations. In particular,
\begin{align*}
    & 0 = \var{01}{L_{01}}{\bar u_0} + \var{12}{L_{12}}{\bar u_2} = \frac{1}{2} \bar w_{t_1} -\frac{1}{2} \bar u_{x} ,  \\
    & 0 = \var{01}{L_{01}}{\bar w_0} + \var{12}{L_{12}}{\bar w_2} = \frac{1}{2} \bar u_{t_1} + \frac{1}{2} x^{2} + \frac{1}{2 \bar w_{x}} , \\
    & 0 = \var{02}{L_{02}}{\bar u_0} - \var{12}{L_{12}}{\bar u_1} = \frac{1}{2} \bar w_{t_2} + \frac{1}{2} x^{2} \bar w_{x} + \frac{1}{4} \bar u_{x}^{2} + x \bar w - \frac{1}{2} \log\left(\bar w_{x}\right) - (\partial_x^{-1} \bar w) , \\
    &  0 = \var{02}{L_{02}}{\bar w_0} - \var{12}{L_{12}}{\bar w_1} = \frac{1}{2} \bar u_{t_2} + \frac{1}{2} x^{2} \bar u_{x} - x \bar u - \frac{\bar u_{x}}{2 \bar w_{x}} .
\end{align*}

Because we do not have an expression for $h_3$, we can only write a single Lagrangian based on the operator $\cB$:
\[ M_{01} =  -\bar u_{xx} \bar w \bar u_t + \log(\bar w_x) \bar w_t + (\partial_x^{-1} \bar w) \bar w_t -x^{2} \bar u_{xx} \bar w + \frac{1}{6} \bar u_{x}^{3} - \log\left(\bar w_{x}\right) \bar u_{x} + 2 \bar u \bar w . \]
We can verify by direct computation that its Euler-Lagrange equations are the CAE equations under the image of $\cB$. Taking the variational derivative with respect to $\bar u$, we find
\begin{align*}
    \var{01}{M_{01}}{\bar u} &= \der{M_{01}}{\bar u} - \partial_x \der{M_{01}}{\bar u_x} - \partial_t \der{M_{01}}{\bar u_t} + \partial_x^2 \der{M_{01}}{\bar u_{xx}} \\
    &= -( \bar w_{xx} + 2 \bar w_x \partial_x ) (\bar u_t + x^2 + \frac{1}{\bar w_x} ) + \bar u_{xx} (\bar w_t - \bar u_x) .
\end{align*}
To deal with the non-local term $(\partial_x^{-1} \bar w) \bar w_t$, the natural extension of the definition of the variational derivative is used. We find
\begin{align*}
    \var{01}{M_{01}}{\bar w} &= - \partial_x^{-1} \der{M_{01}}{(\partial_x^{-1} \bar w)} + \der{M_{01}}{\bar w} - \partial_x \der{L_{01}}{\bar w_x} - \partial_t \der{M_{01}}{\bar w_t}  \\
    &= -( \bar u_{xx} ) \left(\bar u_t + x^2 + \frac{1}{\bar w_x} \right) - \left( \frac{2}{\bar w_x}\partial_x - \frac{\bar w_{xx}}{\bar w_x^2} + 2 \partial_x^{-1} \right) (\bar w_t - \bar u_x) .
\end{align*}

\paragraph{Lagrangian multiform for the scalar CAE}

Finally, we present a Lagrangian multiform for scalar formulation of the CAE. This is not an application of our general construction, but the result is interesting in its own right. Following the results by Pavlov and Zykov in \cite{pavlov2013lagrangian}, we first introduce the transformation of variables
$u=\varphi_{xx}$
so that equation \eqref{constast} is written as
\[
    \varphi_{xxtt}+\left(\frac{1}{\varphi_{xx}}\right)_{xx}+2=0,
\]
which is the Euler-Lagrange equation for the Lagrangian density
\[ 
    L(\varphi,\varphi_{xx},\varphi_{tt})=\frac{1}{2}\varphi_{xx}\varphi_{tt}+\log{\varphi_{xx}}+2\varphi .
\]

Inspired on the latter and similar expressions from \cite{pavlov2013lagrangian}, we build a multiform for the scalar version of the constant astigmatism equation and one of its symmetries. We start with
\begin{align*}
    & L_{01} =  \frac{1}{2} \varphi_{xt_1}^{2} + \log\left(\varphi_{xx}\right) + 2 \varphi \\
    & L_{02} =  \log\left(\varphi_{xx}\right) \varphi_{xt_2} + 2 \sqrt{-{\left(\varphi_{t_2}^{2} - 1\right)} \varphi_{xx}} .
\end{align*}
Here we implicitly assume that $-{\left(\varphi_{t_2}^{2} - 1\right)} \varphi_{xx} > 0$, which is the case for example if $u = \varphi_{xx}$ is positive and $\varphi_{t_2}$ sufficiently small. It is plausible that a suitable branch of the square root can be defined for a larger class of $\varphi$, but here we will focus on formal computations.

We then look for $L_{12}$ such that $\dd \cL$ factorises. We compute
\begin{align*}
    &\partial_{2} L_{01} - \partial_{1} L_{02} 
    = {\left(2 \sqrt{-\frac{\varphi_{xx}}{\varphi_{t_2}^{2} - 1}} \varphi_{t_2} + \frac{\varphi_{xxx}}{\varphi_{xx}} - \varphi_{xxt_1}\right)} {\left(\sqrt{-\frac{\varphi_{t_2}^{2} - 1}{\varphi_{xx}}} + \frac{\varphi_{xt_2}}{\varphi_{xx}} + \varphi_{t_1t_2}\right)}
    - \partial_x L_{12} ,
\end{align*}
where 
\[ L_{12} =  \log\left(\varphi_{xx}\right) \varphi_{t_1t_2} - \varphi_{xt_1} \varphi_{t_1t_2} - 2 \sqrt{-\frac{\varphi_{t_2}^{2} - 1}{\varphi_{xx}}} - \frac{\varphi_{xt_2}}{\varphi_{xx}} . \]
Hence, we obtain a Lagrangian multiform $\cL = L_{01} \dd x \wedge \dd t_1 + L_{02} \dd x \wedge \dd t_2 +  L_{12} \dd t_1 \wedge \dd t_2$ such that all solutions of the system
\begin{subequations}
\label{CAE-scalar}
\begin{align}
    &\varphi_{t_1t_2} = -\sqrt{-\frac{\varphi_{t_2}^{2} - 1}{\varphi_{xx}}} - \frac{\varphi_{xt_2}}{\varphi_{xx}}  \\
    &\varphi_{xxt_1} = 2 \sqrt{-\frac{\varphi_{xx}}{\varphi_{t_2}^{2} - 1}} \varphi_{t_2} + \frac{\varphi_{xxx}}{\varphi_{xx}}
\end{align}
\end{subequations}
are critical points. In particular, this system implies the Euler-Lagrange equation
\[ 0 = \var{12}{L_{12}}{\varphi} = \partial_x^2 \left( \frac{1}{\varphi_{xx}}  \right) + \varphi_{xxt_1t_1} + 2,
\]
which is the constant astigmatism equation for $u = \varphi_{xx}$. Also part of the set of multiform Euler-Lagrange equations, and hence a consequence of equations \eqref{CAE-scalar}, is the surprisingly simple PDE
\[
0 = \var{12}{L_{12}}{\varphi_x} = -\frac{\varphi_{xxt_2}}{\varphi_{xx}^{2}} + \varphi_{t_1t_1t_2} .
\]
This example is similar to several presented in \cite{sleigh2022semidiscrete, ferapontov2025lagrangian} in that the two ways of characterising critical points, the multiform Euler-Lagrange equations and the double-zero expansion of $\dd \cL$, are not obviously equivalent. Hence, in these cases, the Lagrangian multiform provides an easy way of proving equivalence of two systems of PDEs.

Note that we do not have a Hamiltonian structure or a recursion operator in this case. So, it is not clear whether this example is part of a hierarchy of symmetries, or whether it is limited to the 3-dimensional multi-time spanned by $(x,u_2,t_3)$.

\section{Conclusions}

In this paper, we have presented a balanced symplectic counterpart to the standard Hamiltonian formalism for integrable PDEs by introducing a specific class of variables, namely Hamiltonian potential variables. The change to these variables maps compatible pairs of Hamiltonian operators into pairs of symplectic operators that are again compatible. We leveraged this result to establish a direct correspondence between the Lenard-Magri scheme to construct integrable hierarchies in the Hamiltonian setting and the Lenard scheme in the symplectic formalism. Different symplectic operators for a given PDE give rise to non-equivalent Lagrangian formulations of this PDE. Traditionally, the Euler-Lagrange equations of the ``higher'' Lagrangians are increasingly weak differential consequences of the original PDE. However, we showed that when these Lagrangians are extended to form Lagrangian multiforms, the corresponding generalised Euler-Lagrange equations produce the PDEs of the hierachy in their evolutionary form, regardless of which symplectic operator is used.

We analysed a range of examples to illustrate our methods and give evidence of their applicability to a wide range of integrable PDEs. These included the celebrated KdV equation and its dispersionless limits, as well as to multi-component systems arising from polytropic gas dynamics. Finally, we examined the Constant Astigmatism Equation, providing for the first time its Lenard scheme in the symplectic framework, as well as a Lagrangian multiform.

In future work, we intend to study the converse of the process presented in this paper. Specifically, we aim to investigate the construction of bi-Hamiltonian systems starting from bi-Lagrangian (and, by extension, bi-symplectic) structures. Based on this, we hope to obtain explicit integrability conditions on the Lagrangians. Preliminary computations in this direction suggest that, while similar procedures remain applicable, they entail distinct technical challenges.  Another challenge that is left to future work is to extend our approach to bi-Hamiltonian systems where neither of the Hamiltonian operators is constant.

\paragraph{Acknowledgements} PV acknowledges the financial support of GNFM of the Istituto Nazionale di Alta Matematica and the hospitality of the Department of Mathematical Sciences of the Loughborough University.  PV is partially funded by the research project Mathematical Methods in Non-Linear Physics (MMNLP) by the Commissione Scientifica Nazionale – Gruppo 4 – Fisica Teorica of the Istituto Nazionale di Fisica Nucleare (INFN), Sezione di Lecce.	

MV acknowledges support form the Engineering and Physical Sciences Research Council (Project reference EP/Y006712/1).

\bibliographystyle{abbrvnat_mv}
\bibliography{Ver-Ver}

@article{baran2009integrability,
  title = {On Integrability of Weingarten Surfaces: A Forgotten Class},
  shorttitle = {On Integrability of Weingarten Surfaces},
  author = {Baran, Hynek and Marvan, Michal},
  year = 2009,
  journal = {J. Phys. A: Math. Theor.},
  volume = {42},
  number = {40},
  pages = {404007},
  issn = {1751-8121},
  doi = {10.1088/1751-8113/42/40/404007},
  urldate = {2026-03-06},
  langid = {english}
}

@article{bianchi1879ricerchea,
  title = {Ricerche sulle superficie elicoidali e sulle superficie a curvatura costante},
  author = {Bianchi, Luigi},
  year = 1879,
  journal = {Ann. Della Scuola Norm. Super. Pisa - Cl. Sci.},
  volume = {2},
  pages = {285--341},
  langid = {italian}
}

@misc{bocharov1999symmetries,
  title = {Symmetries and Conservation Laws for Differential Equations of Mathematical Physics},
  author = {Bocharov, A. and Chetverikov, V. and Duzhin, S. and Khor’kova, N. and Samokhin, A. and Torkhov, Yu and Verbovetsky, A.},
  year = 1999,
  journal = {American Mathematical Society},
  series = {Translations of Mathematical Monographs},
  volume = {182},
  publisher = {American Mathematical Society},
  issn = {0065-9282, 2472-5137},
  doi = {10.1090/mmono/182},
  urldate = {2026-03-06},
  howpublished = {http://www.ams.org/mmono/182},
  isbn = {9780821809587 9780821897898 9781470445966},
  langid = {english}
}

@article{bustamante2003multilagrangians,
  title = {Multi-Lagrangians, Hereditary Operators and Lax Pairs for the Korteweg–de Vries Positive and Negative Hierarchies},
  author = {Bustamante, Miguel D. and Hojman, Sergio A.},
  year = 2003,
  journal = {J. Math. Phys.},
  volume = {44},
  number = {10},
  pages = {4652--4671},
  issn = {0022-2488},
  doi = {10.1063/1.1609035},
  urldate = {2025-10-08}
}

@article{caudrelier2020hamiltonian,
  title = {Hamiltonian Multiform Description of an Integrable Hierarchy},
  author = {Caudrelier, Vincent and Stoppato, Matteo},
  year = 2020,
  journal = {Journal of Mathematical Physics},
  volume = {61},
  number = {12},
  pages = {123506},
  issn = {0022-2488},
  doi = {10.1063/5.0012153},
  urldate = {2024-02-06}
}

@article{caudrelier2021multiform,
  title = {Multiform Description of the AKNS Hierarchy and Classical R-Matrix},
  author = {Caudrelier, Vincent and Stoppato, Matteo},
  year = 2021,
  journal = {J. Phys. A: Math. Theor.},
  volume = {54},
  number = {23},
  pages = {235204},
  publisher = {IOP Publishing},
  issn = {1751-8121},
  doi = {10.1088/1751-8121/abfac9},
  urldate = {2024-02-06},
  langid = {english}
}

@article{caudrelier2025geometry,
  title = {On the Geometry of Lagrangian One-Forms},
  author = {Caudrelier, Vincent and Harland, Derek},
  year = 2025,
  journal = {Lett Math Phys},
  volume = {115},
  number = {2},
  pages = {38},
  issn = {1573-0530},
  doi = {10.1007/s11005-025-01925-0},
  urldate = {2026-03-30},
  langid = {english}
}

@article{dealmeidadasilva1990simple,
  title = {A Simple Lagrangian for Integrable Systems},
  author = {{\noopsort{almeida da silva}}{de Almeida da Silva}, M. A. and Das, Ashok},
  year = 1990,
  journal = {J. Math. Phys.},
  volume = {31},
  number = {4},
  pages = {798--800},
  issn = {0022-2488},
  doi = {10.1063/1.528813},
  urldate = {2026-03-06}
}

@book{dorfman1993dirac,
  title = {Dirac Structures and Integrability of Nonlinear Evolution Equations},
  author = {Dorfman, Irene},
  year = 1993,
  publisher = {Wiley \& Sons},
  address = {Chichester, England},
  isbn = {978-0-471-93893-4}
}

@article{dubrovin1983hamiltonianformalism,
  title = {Hamiltonian-Formalism of One-Dimensional Systems of the Hydrodynamic Type and the Bogolyubov-Whitham Averaging Method},
  author = {Dubrovin, Ba and Novikov, Sp},
  year = 1983,
  journal = {Dokl. Akad. Nauk Sssr},
  volume = {270},
  number = {4},
  pages = {781--785},
  publisher = {Mezhdunarodnaya Kniga},
  address = {Moscow},
  issn = {0002-3264},
  langid = {russian}
}

@article{ferapontov1992nonlocal,
  title = {Nonlocal Matrix Hamiltonian Operators, Differential Geometry, and Applications},
  author = {Ferapontov, E. V.},
  year = 1992,
  journal = {Theor Math Phys},
  volume = {91},
  number = {3},
  pages = {642--649},
  issn = {1573-9333},
  doi = {10.1007/BF01017341},
  urldate = {2026-03-06},
  langid = {english}
}

@article{ferapontov2006class,
  title = {On a Class of Three-Dimensional Integrable Lagrangians},
  author = {Ferapontov, E.V. and Khusnutdinova, K.R. and Tsarev, S.P.},
  year = 2006,
  journal = {Commun. Math. Phys.},
  volume = {261},
  number = {1},
  pages = {225--243},
  issn = {1432-0916},
  doi = {10.1007/s00220-005-1415-5},
  urldate = {2026-03-06},
  langid = {english}
}

@article{ferapontov2010integrable,
  title = {Integrable Lagrangians and Modular Forms},
  author = {Ferapontov, E. V. and Odesskii, A. V.},
  year = 2010,
  journal = {Journal of Geometry and Physics},
  volume = {60},
  number = {6},
  pages = {896--906},
  issn = {0393-0440},
  doi = {10.1016/j.geomphys.2010.02.006},
  urldate = {2026-03-06}
}

@article{ferapontov2010integrablea,
  title = {Integrable Equations of the Dispersionless Hirota Type and Hypersurfaces in the Lagrangian Grassmannian},
  author = {Ferapontov, Evgeny Vladimirovich and Hadjikos, Lenos and Khusnutdinova, Karima Robertovna},
  year = 2010,
  journal = {Int Math Res Notices},
  volume = {2010},
  number = {3},
  pages = {496--535},
  issn = {1073-7928},
  doi = {10.1093/imrn/rnp134},
  urldate = {2026-03-06}
}

@article{ferapontov2025lagrangian,
  title = {Lagrangian Multiforms and Dispersionless Integrable Systems},
  author = {Ferapontov, Evgeny V. and Vermeeren, Mats},
  year = 2025,
  journal = {Lett Math Phys},
  volume = {115},
  number = {6},
  pages = {125},
  issn = {1573-0530},
  doi = {10.1007/s11005-025-02016-w},
  urldate = {2026-03-06},
  langid = {english}
}

@article{hlavac2014reciprocal,
  title = {A Reciprocal Transformation for the Constant Astigmatism Equation},
  author = {Hlaváč, Adam and Marvan, Michal},
  year = 2014,
  journal = {SIGMA Symmetry Integrability Geom. Methods Appl.},
  volume = {10},
  pages = {091},
  publisher = {{SIGMA. Symmetry, Integrability and Geometry: Methods and Applications}},
  issn = {18150659},
  doi = {10.3842/SIGMA.2014.091},
  urldate = {2026-03-06},
  langid = {english}
}

@article{hlavac2015multisoliton,
  title = {On Multisoliton Solutions of the Constant Astigmatism Equation},
  author = {Hlaváč, Adam},
  year = 2015,
  journal = {J. Phys. A: Math. Theor.},
  volume = {48},
  number = {36},
  pages = {365202},
  publisher = {IOP Publishing},
  issn = {1751-8121},
  doi = {10.1088/1751-8113/48/36/365202},
  urldate = {2026-03-06},
  langid = {english}
}

@article{hlavac2017nonlocal,
  title = {Nonlocal Conservation Laws of the Constant Astigmatism Equation},
  author = {Hlaváč, Adam and Marvan, Michal},
  year = 2017,
  journal = {Journal of Geometry and Physics},
  series = {Nonlinear Partial Differential Equations: Integrability, Geometry and Related Topics},
  volume = {113},
  pages = {117--130},
  issn = {0393-0440},
  doi = {10.1016/j.geomphys.2016.06.002},
  urldate = {2026-03-06}
}

@article{hlavac2018more,
  title = {More Exact Solutions of the Constant Astigmatism Equation},
  author = {Hlaváč, Adam},
  year = 2018,
  journal = {Journal of Geometry and Physics},
  volume = {123},
  pages = {209--220},
  issn = {0393-0440},
  doi = {10.1016/j.geomphys.2017.09.003},
  urldate = {2026-03-06}
}

@article{krasilshchik2011geometry,
  title = {Geometry of Jet Spaces and Integrable Systems},
  author = {Krasil’shchik, Joseph and Verbovetsky, Alexander},
  year = 2011,
  journal = {Journal of Geometry and Physics},
  series = {The Interface between Integrability and Quantization},
  volume = {61},
  number = {9},
  pages = {1633--1674},
  issn = {0393-0440},
  doi = {10.1016/j.geomphys.2010.10.012},
  urldate = {2026-03-06}
}

@article{lipschitz1900zur,
  title = {Zur Theorie Der Krummen Oberflächen},
  author = {Lipschitz, R.},
  year = 1900,
  journal = {Acta Math.},
  volume = {10},
  number = {none},
  pages = {131--136},
  publisher = {Institut Mittag-Leffler},
  issn = {0001-5962, 1871-2509},
  doi = {10.1007/BF02393697},
  urldate = {2026-03-06},
  langid = {english}
}

@article{lobb2009lagrangian,
  title = {Lagrangian Multiforms and Multidimensional Consistency},
  author = {Lobb, Sarah and Nijhoff, Frank},
  year = 2009,
  journal = {J. Phys. Math. Theor.},
  volume = {42},
  pages = {454013},
  doi = {10.1088/1751-8113/42/45/454013}
}

@article{magri1978simple,
  title = {A Simple Model of the Integrable Hamiltonian Equation},
  author = {Magri, Franco},
  year = 1978,
  journal = {J. Math. Phys.},
  volume = {19},
  number = {5},
  pages = {1156--1162},
  issn = {0022-2488},
  doi = {10.1063/1.523777},
  urldate = {2026-03-06}
}

@article{manganaro2014constant,
  title = {The Constant Astigmatism Equation. New Exact Solution},
  author = {Manganaro, N and Pavlov, M V},
  year = 2014,
  journal = {J. Phys. A: Math. Theor.},
  volume = {47},
  number = {7},
  pages = {075203},
  publisher = {IOP Publishing},
  issn = {1751-8121},
  doi = {10.1088/1751-8113/47/7/075203},
  urldate = {2026-03-06},
  langid = {english}
}

@article{mokhov1998symplectic,
  title = {Symplectic and Poisson Structures on Loop Spaces of Smooth Manifolds, and Integrable Systems},
  author = {Mokhov, O. I.},
  year = 1998,
  journal = {Russ. Math. Surv.},
  volume = {53},
  number = {3},
  pages = {515},
  publisher = {IOP Publishing},
  issn = {0036-0279},
  doi = {10.1070/RM1998v053n03ABEH000019},
  urldate = {2025-11-06},
  langid = {english}
}

@book{mokhov2001symplectic,
  title = {Symplectic and Poisson Geometry on Loop Spaces of Smooth Manifolds and Integrable Equations},
  author = {Mokhov, O. I.},
  year = 2001,
  publisher = {Harwood Academic Publishers},
  address = {Amsterdam},
  isbn = {978-90-5823-235-9},
  langid = {english}
}

@article{nutku1987new,
  title = {On a New Class of Completely Integrable Nonlinear Wave Equations. II. Multi‐Hamiltonian Structure},
  author = {Nutku, Y.},
  year = 1987,
  journal = {J. Math. Phys.},
  volume = {28},
  number = {11},
  pages = {2579--2585},
  issn = {0022-2488},
  doi = {10.1063/1.527749},
  urldate = {2026-03-13}
}

@incollection{nutku2001lagrangian,
  title = {Lagrangian Approach to Integrable Systems Yields New Symplectic Structures for KDV},
  booktitle = {Integrable Hierarchies and Modern Physical Theories},
  author = {Nutku, Y.},
  editor = {Aratyn, Henrik and Sorin, Alexander S.},
  year = 2001,
  pages = {203--213},
  publisher = {Springer Netherlands},
  address = {Dordrecht},
  doi = {10.1007/978-94-010-0720-7_6},
  urldate = {2026-03-06},
  isbn = {978-94-010-0720-7},
  langid = {english}
}

@article{nutku2002multilagrangians,
  title = {Multi-Lagrangians for Integrable Systems},
  author = {Nutku, Y. and Pavlov, M. V.},
  year = 2002,
  journal = {J. Math. Phys.},
  volume = {43},
  number = {3},
  pages = {1441--1459},
  issn = {0022-2488},
  doi = {10.1063/1.1427765},
  urldate = {2026-03-06}
}

@book{olver1986applications,
  title = {Applications of Lie Groups to Differential Equations},
  author = {Olver, Peter J.},
  year = 1986,
  series = {Graduate Texts in Mathematics},
  volume = {107},
  publisher = {Springer},
  address = {New York, NY},
  doi = {10.1007/978-1-4684-0274-2},
  urldate = {2024-02-06},
  isbn = {978-1-4684-0276-6 978-1-4684-0274-2},
  langid = {english}
}

@article{olver1988hamiltonian,
  title = {Hamiltonian Structures for Systems of Hyperbolic Conservation Laws},
  author = {Olver, Peter J. and Nutku, Yavuz},
  year = 1988,
  journal = {J. Math. Phys.},
  volume = {29},
  number = {7},
  pages = {1610--1619},
  issn = {0022-2488},
  doi = {10.1063/1.527909},
  urldate = {2025-09-24}
}

@article{pavlov2013lagrangian,
  title = {Lagrangian and Hamiltonian Structures for the Constant Astigmatism Equation},
  author = {Pavlov, Maxim V and Zykov, Sergej A},
  year = 2013,
  journal = {J. Phys. A: Math. Theor.},
  volume = {46},
  number = {39},
  pages = {395203},
  publisher = {IOP Publishing},
  issn = {1751-8121},
  doi = {10.1088/1751-8113/46/39/395203},
  urldate = {2026-03-06},
  langid = {english}
}

@article{pavlov2017remarks,
  title = {Remarks on the Lagrangian Representation of Bi-Hamiltonian Equations},
  author = {Pavlov, M. V. and Vitolo, R. F.},
  year = 2017,
  journal = {Journal of Geometry and Physics},
  series = {Nonlinear Partial Differential Equations: Integrability, Geometry and Related Topics},
  volume = {113},
  pages = {239--249},
  issn = {0393-0440},
  doi = {10.1016/j.geomphys.2016.10.013},
  urldate = {2026-09-09}
}

@article{pavlov2021classification,
  title = {Classification of Bi-Hamiltonian Pairs Extended by Isometries},
  author = {Pavlov, Maxim V. and Vergallo, Pierandrea and Vitolo, Raffaele},
  year = 2021,
  journal = {Proc. A},
  volume = {477},
  number = {2251},
  pages = {20210185},
  issn = {1364-5021},
  doi = {10.1098/rspa.2021.0185},
  urldate = {2026-03-06}
}

@article{petrera2017variational,
  title = {Variational Symmetries and Pluri-Lagrangian Systems in Classical Mechanics},
  author = {Petrera, Matteo and Suris, {\relax Yu}ri B},
  year = 2017,
  journal = {J. Nonlinear Math. Phys.},
  volume = {24 (Sup. 1)},
  pages = {121--145},
  doi = {10.1080/14029251.2017.1418058}
}

@article{petrera2021variational,
  title = {Variational Symmetries and Pluri-Lagrangian Structures for Integrable Hierarchies of PDEs},
  author = {Petrera, Matteo and Vermeeren, Mats},
  year = 2021,
  journal = {Eur. J. Math.},
  volume = {7},
  pages = {741--765},
  publisher = {Springer},
  doi = {10.1007/s40879-020-00436-7}
}

@book{saunders1989geometry,
  title = {The Geometry of Jet Bundles},
  author = {Saunders, D. J.},
  year = 1989,
  series = {London Mathematical Society Lecture Note Series},
  publisher = {Cambridge University Press},
  address = {Cambridge},
  doi = {10.1017/CBO9780511526411},
  urldate = {2026-03-06},
  isbn = {978-0-521-36948-0}
}

@article{sleigh2020variational,
  title = {Variational Symmetries and Lagrangian Multiforms},
  author = {Sleigh, Duncan and Nijhoff, Frank and Caudrelier, Vincent},
  year = 2020,
  journal = {Lett Math Phys},
  volume = {110},
  number = {4},
  pages = {805--826},
  issn = {1573-0530},
  doi = {10.1007/s11005-019-01240-5},
  urldate = {2023-12-18},
  langid = {english}
}

@article{sleigh2022semidiscrete,
  title = {Semi-Discrete Lagrangian 2-Forms and the Toda Hierarchy},
  author = {Sleigh, Duncan and Vermeeren, Mats},
  year = 2022,
  journal = {J. Phys. A.},
  volume = {55},
  number = {47},
  pages = {475204},
  publisher = {IOP Publishing},
  issn = {1751-8121},
  doi = {10.1088/1751-8121/aca451},
  urldate = {2023-01-28},
  langid = {english}
}

@incollection{suris2016lagrangian,
  title = {On the Lagrangian Structure of Integrable Hierarchies},
  booktitle = {Advances in Discrete Differential Geometry},
  author = {Suris, {\relax Yu}ri B and Vermeeren, Mats},
  year = 2016,
  pages = {347--378},
  publisher = {Springer},
  address = {Berlin, Heidelberg},
  doi = {10.1007/978-3-662-50447-5_11}
}

@incollection{suris2016variational,
  title = {Variational Symmetries and Pluri-Lagrangian Systems},
  booktitle = {Dynamical Systems, Number Theory and Applications: A Festschrift in Honor of Armin Leutbecher's 80th Birthday},
  author = {Suris, {\relax Yu}ri B},
  year = 2016,
  pages = {255--266},
  publisher = {World Scientific},
  address = {New Jersey, etc.},
  doi = {10.1142/9789814699877_0013}
}

@article{vermeeren2019continuum,
  title = {Continuum Limits of Pluri-Lagrangian Systems},
  author = {Vermeeren, Mats},
  year = 2019,
  journal = {J. Integrable Syst.},
  volume = {4},
  number = {1},
  pages = {xyy020},
  publisher = {Oxford University Press},
  doi = {10.1093/integr/xyy020}
}

@article{vermeeren2021hamiltonian,
  title = {Hamiltonian Structures for Integrable Hierarchies of Lagrangian PDEs},
  author = {Vermeeren, Mats},
  year = 2021,
  journal = {Open Commun. Nonlinear Math. Phys.},
  volume = {1},
  pages = {ocnmp:7491},
  doi = {10.46298/ocnmp.7491}
}

@preamble{ "\providecommand{\noopsort}[1]{} " }

\end{document}